\documentclass[11pt]{scrartcl}
 \pdfoutput=1 

\usepackage[a4paper, total={16cm, 24cm}]{geometry}

\usepackage[hyphens]{url}  \usepackage{hyperref}
\hypersetup{
     colorlinks   = true,
     urlcolor    = teal,
	 citecolor = blue,
	 linkcolor = teal
}

\usepackage[utf8]{inputenc}
\usepackage{amssymb,latexsym}
\usepackage{amsmath,amsfonts}
\usepackage{natbib}

\usepackage{microtype}
\usepackage{tikz}
\usepackage{paralist}
\usepackage{color,xspace}
\usepackage{tabularx}
\usepackage{nicefrac}
\usepackage{algorithm}
\usepackage{algpseudocode}
\usepackage{booktabs}

\usepackage{pifont}
\newcommand{\cmark}{\ding{51}}
\newcommand{\xmark}{-}

\usepackage{amsthm}
\usepackage{mathtools}
\usepackage{thm-restate}

\usepackage[capitalize,nameinlink]{cleveref}

\newtheorem{theorem}{Theorem}[section]

\newtheorem{proposition}{Proposition}[section]
\newtheorem{example}{Example}[section]
\theoremstyle{definition}
\newtheorem{definition}{Definition}[section]

\theoremstyle{remark}
\newtheorem{remarkfixed}{Remark}[section]

\newcommand{\ie}{i.e.,\xspace}
\newcommand{\eg}{e.g.,\xspace}

\newcommand{\MC}{\Delta}
\newcommand{\PAV}{\mathrm{PAV}}

\renewcommand*{\le}{\leqslant}
\renewcommand*{\ge}{\geqslant}
\renewcommand*{\leq}{\leqslant}
\renewcommand*{\geq}{\geqslant}

\newcommand{\problemdef}[3]{
\begin{center}
	\begin{minipage}{0.9\textwidth}
		\noindent
		\textsc{#1}

				\vspace{2pt}
				\setlength{\tabcolsep}{3pt}
				\begin{tabularx}{\textwidth}{@{}lX@{}}
						\textbf{Input:} 		& #2 \\
						\textbf{Question:} 	& #3
					\end{tabularx}
	\end{minipage}
		\end{center}
}

\usepackage[colorinlistoftodos,textsize=small]{todonotes} 
\newcommand{\N}{N} \newcommand{\C}{C} \newcommand{\W}{W} \newcommand{\A}{A} \renewcommand{\l}{\ell}
\renewcommand{\P}{P} 

\newcommand{\pa}{party-approval\xspace}
\newcommand{\ca}{candidate-approval\xspace}

\newcommand{\phrag}{Phragm\'{e}n\xspace}
\newcommand{\maxP}{leximax-\phrag} \newcommand{\varP}{var-\phrag}
\newcommand{\seqP}{seq-\phrag}
\newcommand{\eneP}{Eneström-\phrag}

\title{Approval-Based Apportionment}

\usepackage{authblk}

\author[1]{Markus Brill}
\author[2]{Paul G\"olz}
\author[3]{Dominik Peters}
\author[1]{\authorcr Ulrike Schmidt-Kraepelin}
\author[1]{Kai Wilker}
\affil[1]{TU Berlin} 
\affil[2]{Carnegie Mellon University}
\affil[3]{CNRS, LAMSADE, Universit\'e Paris-Daupine--PSL}
	
\date{\vspace{-1cm}}

\begin{document}

\maketitle

\begin{abstract}
In the apportionment problem, a fixed number of seats must be distributed among parties in proportion to the number of voters supporting each party. We study a generalization of this setting, in which voters can support multiple parties by casting approval ballots. This approval-based apportionment setting generalizes traditional apportionment and is a natural restriction of approval-based multiwinner elections, where approval ballots range over individual candidates instead of parties. Using techniques from both apportionment and multiwinner elections, we identify rules that generalize the D'Hondt apportionment method and that satisfy strong axioms which are generalizations of properties commonly studied in the apportionment literature. In fact, the rules we discuss provide representation guarantees that are currently out of reach in the general setting of multiwinner elections: First, we show that core-stable committees are guaranteed to exist and can be found in polynomial time. Second, we demonstrate that extended justified representation is compatible with committee monotonicity (also known as house monotonicity).
\end{abstract}

\section{Introduction}
\label{sec:intro}

The fundamental fairness principle of \textit{proportional representation} is relevant in a variety of applications ranging from recommender systems to digital democracy~\citep{Bril21a}. 
It features most explicitly in the context of political elections, which is the language we adopt for this paper.   
In this context, proportional representation prescribes that the number of representatives championing an opinion in a legislature should be proportional to the number of voters who favor that opinion. 

In most democratic institutions, proportional representation is implemented via what we call \emph{party-choice elections}: Candidates are members of political parties and voters are asked to choose their favorite party; 
each party is then allocated a number of seats that is (approximately) proportional to the number of votes it received. The problem of transforming a voting outcome into a distribution of seats is known as \emph{apportionment}. Analyzing the advantages and disadvantages of different apportionment methods has a long and illustrious political history and has given rise to an elegant mathematical theory~\citep{BaYo82a,Puke14a}.
\looseness=-1

Forcing voters to choose a single party prevents them from communicating any preferences beyond their most preferred alternative. For example, if a voter feels equally well represented by several political parties, there is no way to express this preference within the voting system. 
In the context of single-winner elections, \emph{approval voting} has been put forward as a solution to this problem as it strikes an attractive compromise between simplicity and expressivity \citep{BrFi07c,LaSa10a}.
Under approval voting, each voter is asked to specify a set of candidates she ``approves of,'' \ie voters can arbitrarily partition the set of candidates into approved candidates and disapproved ones. 
Proponents of approval voting argue that its introduction could 
increase voter turnout,
``help elect the strongest candidate,'' and 
``add legitimacy to the outcome'' of an election \citep[pp. 4--8]{BrFi07c}. 

The practical and theoretical appeal of approval voting in single-winner elections has led a number of scholars to suggest to also use approval voting for multiwinner elections, in which a fixed number of candidates need to be elected \citep{KiMa12a}.
Whereas, in the single-winner setting, the straightforward voting rule ``choose the candidate approved by the highest number of voters'' enjoys a strong axiomatic foundation \citep{Fish78d,Fish79a,Alos06a},
several ways of aggregating approval ballots have been proposed for the multiwinner setting  \citep{KiMa12a,LaSk20v3}.

Most studies of approval-based multiwinner elections assume that voters directly express their preference over individual candidates; we refer to this setting as \emph{candidate-approval} elections. This assumption runs counter to widespread democratic practice, in which candidates belong to political parties and voters indicate preferences over these parties (which induce implicit preferences over candidates).
In this paper, we therefore study \emph{party-approval} elections, in which voters express approval votes over parties and a given number of seats must be distributed among the parties.
We refer to the process of allocating these seats as \emph{approval-based apportionment}.

Throughout this paper, we interpret a ballot that approves a set~$S$ of parties as a preference for legislatures with a larger total number of members from parties in~$S$.
This interpretation generalizes the natural interpretation of party-choice ballots as preferences for legislatures with a larger number of members of the chosen party.
Our interpretation implicitly imputes perfect indifference between approved parties.
This means that we assume voters to be indifferent to the distribution of seats between approved parties. For example, consider only legislatures with a fixed total number of seats given to approved parties. Then a voter would be indifferent between a legislature where the approved parties all get an equal number of seats, and a legislature where just one of the approved parties obtains all those seats.
While this assumption is restrictive, it does allow for a simple voting process, and the additional expressivity of approval ballots compared to party-choice ballots seems attractive.

Indeed, we believe that party-approval elections are a promising framework for legislative elections in the real world, especially since allowing voters to approve multiple parties enables the aggregation mechanism to coordinate like-minded voters.
For example, under party-choice elections, two groups of voters might vote for parties that they mutually disapprove of. Approval ballots could reveal that both groups approve a third party of more general appeal. Given this information, a voting rule could then allocate more seats to this third party, leading to mutual gain. This cooperation is particularly necessary for small minority opinions that are not centrally coordinated. In such cases, finding a commonly approved party can make the difference between being represented or votes being wasted because the individual parties receive insufficient support.

One aspect that makes it easier to transition from party-choice elections to \pa elections (rather than to \ca elections) is that \pa elections can be implemented as \emph{closed-list} systems.
That is, parties can retain the power to choose the ordering in which their candidates are allocated seats, as they do in many current democratic systems.
By contrast, \ca elections necessarily confer this power to the voters (leading to an \emph{open-list} system), which might give parties an incentive to oppose a change of the voting system.
Of course, \pa elections are compatible with an open-list approach, since we can run a secondary mechanism alongside the \pa election to determine the order of party candidates.

\subsection{Related Work}

To the best of our knowledge, this paper is the first to formally develop and systematically study approval-based apportionment.
That said, several scholars have previously explored possible generalizations of existing aggregation procedures to allow for approval votes over parties.

For instance, \citet{BKP19a} 
study multiwinner approval rules that are inspired by classical apportionment methods.
Besides the setting of candidate approval, they explicitly consider the case where voters cast \pa votes. They conclude that these rules could
``encourage coalitions across party or factional lines, thereby diminishing gridlock and promoting consensus.''

Such desire for compromise is only one motivation for considering party-approval elections, as exemplified by recent work by \citet{SpGe19a}.
To allow for more efficient governing, they aim to concentrate the power of a legislature in the hands of few big parties, while nonetheless preserving the principle of proportional representation.
To this end, they let voters cast \pa votes and transform these votes into a party-choice election by assigning each voter to one of her approved parties.
Specifically, they propose to assign
voters to parties so that the strongest party has as many votes as possible. We later call this method \emph{majoritarian portioning}.

Several other papers consider extensions of approval-based voting rules to accommodate \pa elections.
In their paper introducing the \textit{satisfaction approval voting} rule, \citet{BrKi14a} discuss a variant of this rule adapted for party-approval votes. 
\citet{MoOl15a} and \citet{CMS19a} study two approval-based multiwinner rules due to Phragm\'en and Enestr\"om, and note that they also work for party-approval elections (which is true for any multiwinner rule using the embedding that we discuss in \Cref{sec:using_capp}). Both papers consider a monotonicity axiom for party-approval elections (``if a party receives additional approvals, it should receive additional seats'') but find that their two methods fail it.
For the case of two parties, they analyze the behavior of these rules as the house size approaches infinity. They find that both rules fail to converge to the most natural seat distribution.
\citet{JO19} analyze the limit behavior in more detail, and also show that Thiele's sequential rule (aka SeqPAV) does converge to the ideal value.

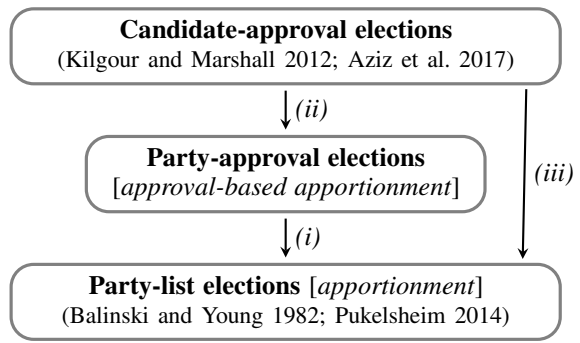
\begin{figure}
	\centering
\begin{tikzpicture}[scale=1, >=stealth,	shorten >=2pt, shorten < = 2pt, pa/.style={
rectangle,minimum height=8mm, minimum width=20mm, rounded corners=3mm, text width=110mm, text centered, inner sep=1.5mm,
very thick,draw=black!50, }]]
      \draw (0,0)      node[pa](ABC)
      {\textbf{Candidate-approval elections} \\ \citep{KiMa12a,LaSk20v3}} 
      ++(0,-2.2)
      node[pa,text width=70mm](PAPP)
      {\textbf{Party-approval elections}\\ 
      {(\textit{approval-based apportionment})} \\
      }  
      ++(0,-2.2)         node[pa](CAPP)
      {\textbf{Party-choice elections}
      {(\textit{apportionment})} \\
      \citep{BaYo82a,Puke14a}};     
      \draw[->,thick] (ABC) to node [auto]{\textit{(i)}} (PAPP);
      \draw[->, thick] (PAPP) to node [auto]{\textit{(ii)}} (CAPP);
      \draw[->, thick] (4.2,-0.59)
      to node [auto]{\textit{(iii)}}
      (4.2,-3.8);
\end{tikzpicture}
\caption{Relations between the different settings of multiwinner elections. An arrow from $X$ to $Y$ signifies that $X$ is a generalization of $Y$.
The relationship corresponding to arrow~\textit{(iii)} has been explored by \protect\citet{BLS18a}. We establish and explore the relationship~\textit{(i)} in \Cref{sec:using_capp} and the relationship~\textit{(ii)} in \Cref{sec:using_app}.}  
\label{fig:relations}
\end{figure}

\subsection{Relation to Other Settings}
\label{sec:relations}

We can position party-approval elections between two well-studied voting settings (see \Cref{fig:relations}). 

First, our setting can be viewed as a \textit{special case of approval-based multiwinner voting}, in which voters cast \textit{\ca} votes.
A \pa election can be embedded in this setting by replacing each party by multiple candidates belonging to this party, and by interpreting a voter's approval of a party as approval of all of its candidates. 
This embedding establishes \pa elections as a subdomain of \ca elections (see arrow~\textit{(i)} in \Cref{fig:relations}). 
In \Cref{sec:using_capp}, we explore the axiomatic and computational ramifications of this domain restriction. 

Second, approval-based apportionment \textit{generalizes standard apportionment} (arrow \textit{(ii)}), which corresponds to \pa elections in which all approval sets are singletons (i.e., party-choice elections). 
In \Cref{sec:using_app}, we propose a method to generalize apportionment methods to the party-approval setting using so-called \emph{portioning} methods.

\subsection{Contributions}
In this paper, we formally introduce the setting of approval-based apportionment and explore different possibilities of constructing axiomatically desirable aggregation methods for this setting.
Besides its conceptual appeal, this setting is also interesting from a technical perspective. 

Exploiting the relations described in \Cref{sec:relations}, we resolve problems that remain open in the more general setting of 
\ca elections. 
First, we show that the core of an approval-based apportionment problem is always nonempty, and that a popular multiwinner rule known as Proportional Approval Voting (PAV) always returns a core-stable committee.
We also present a polynomial-time variant of PAV that is also core stable.
Second, we prove that committee monotonicity is compatible with extended justified representation (a representation axiom proposed by \citet{ABC+16a}) by providing a rule that satisfies both properties.

Some familiar multiwinner rules (in particular, PAV) provide stronger representation guarantees when applied in the party-approval setting. However, for many standard multiwinner voting rules, we give examples that show that their axiomatic guarantees do not improve in the party-approval setting. From a computational complexity perspective, we show that some rules known to be NP-hard in the candidate-approval setting remain NP-hard to evaluate in the party-approval setting.
However, it becomes computationally easier to reason about proportionality axioms. Specifically, we show that it is tractable to check whether a given committee satisfies extended justified representation (or the weaker axiom of proportional justified representation). The analogs of these problems for candidate-approval elections are coNP-hard. These tractability results do not extend to checking whether a committee is core-stable: we show that this problem is coNP-complete for both party-approval and candidate-approval elections.

\section{The Model}
\label{sec:papp}

A \emph{\pa election} is a tuple $(\N,\P,\A,k)$ consisting of a set of voters $\N=\{1,\ldots,n\}$, a finite set of parties $\P$, a ballot profile $\A=(A_1,\ldots,A_n)$ where each ballot $A_i \subseteq \P$ is the set of parties approved by voter $i$, and the committee size $k \in \mathbb{N}$.
We assume that $A_i \neq \emptyset$ for all $i \in N$. When considering computational problems, we assume that $k$ is encoded in unary (see Remark~\ref{remark:unary}). This is a mild restriction since in most applications (such as legislative elections), $k$ is smaller than the number of voters.

A \emph{committee} in this setting is a multiset $\W : \P \rightarrow \mathbb{N}$ over parties, which determines the number of seats $\W(p)$ assigned to each party $p \in \P$. The size of a committee $W$ is $|W|=\sum_{p \in P} W(p)$, and we denote multiset addition and subtraction by $+$ and $-$, respectively.
For a voter $i$ and a committee $W$, we write $u_i(W) = \sum_{p \in A_i} W(p)$ for the number of seats in $W$ that are allocated to parties approved by voter $i$.
A \textit{\pa rule} is a function that takes a \pa election $(\N,\P,\A,k)$ as input and returns a committee $\W$ of valid size $|\W| = k$.\footnote{This definition implies that rules are \textit{resolute}, i.e., they only return a single committee. In the case of a tie between multiple committees, a tiebreaking mechanism is necessary. Our results hold independently of the choice of a specific tiebreaking mechanism.}

In our axiomatic study of \pa rules, we focus on two axioms capturing proportional representation:
extended justified representation and core stability.
Both are derived from their analogs in \ca elections (see \Cref{axiomsderivation}) where they were proposed by \citet{ABC+16a}.
To state these axioms, it is helpful to define the \emph{quota} of a subset $S$ of voters as
$q(S) = \lfloor k \cdot |S| / n \rfloor$.
Intuitively, $q(S)$ corresponds to the number of seats that the group $S$ ``deserves'' to be represented by (rounded down). 

\begin{definition}\label{def:ejr}
A committee $\W : \P \rightarrow \mathbb{N}$ provides \emph{extended justified representation (EJR)} for a \pa election $(\N,\P,\A,k)$  if there is no subset $S \subseteq \N$ of voters such that
$\bigcap_{i \in S} A_i \neq \emptyset$
and $u_i(W) < q(S)$ for all $i \in S$.
\end{definition}

In words, EJR requires that for every voter group $S$ with a commonly approved party, at least one voter of the group must approve at least $q(S)$ committee members. 
A \pa rule is said to \emph{satisfy EJR} 
if it only produces committees providing~EJR.

We can obtain a stronger representation axiom by removing the requirement of a commonly approved party.

\begin{definition}\label{def:core}
A committee $\W: \P \to \mathbb{N}$ is \emph{core stable} for a \pa election $(\N,\P,\A,k)$ if there is no nonempty subset $S \subseteq \N$ and committee $T : \P \to \mathbb{N}$ of size $|T| \leq q(S)$ such that $u_i(T) > u_i(W)$
for all $i \in S$.
The \emph{core} of a \pa election is the set of all core-stable committees.
\end{definition}

Core stability requires adequate representation even for voter groups that cannot agree on a common party, by ruling out the possibility that the group can deviate to a smaller committee that represents all voters in the group strictly better.
It follows from the definitions that core stability is a stronger requirement than EJR: 
If a committee violates EJR, there is a group $S$ that would prefer any committee of size $q(S)$ that assigns all seats to the commonly approved party.

Besides these representation axioms, a final axiom that we will discuss is \emph{committee monotonicity} \citep[\eg][]{BaCo08a,EFSS17a}.
A \pa rule $f$ satisfies this axiom if, for all \pa elections $(\N, \P, \A, k)$, it holds that $f(\N, \P, \A, k) \subseteq f(\N, \P, \A, k+1)$.
The apportionment literature calls this \textit{house monotonicity}.
Committee monotonic rules avoid the so-called \emph{Alabama paradox}, in which a party loses a seat when the committee size increases.
They can also be used to construct proportional rankings~\citep{SLB+17a,IsBr21b}.

\section{Constructing Party-Approval Rules via Multiwinner Voting Rules}
\label{sec:using_capp}

In this section, we show how \pa elections can be translated into \ca elections. This embedding allows us to apply established \ca rules to our setting. Exploiting this fact, we will prove the existence of core-stable committees for \pa elections.

\subsection{Preliminaries}
\label{sec:prelims:CAPP}

A \emph{candidate-approval election} is a tuple $(\N, \C, \A, k)$.
Just as for \pa elections, $\N = \{1, \dots, n\}$ is a set of voters, $\C$ is a finite set, $\A$ is an $n$-tuple of nonempty subsets of $\C$, and $k \in \mathbb{N}$ is the committee size.
The conceptual difference is that $\C$ is a set of individual candidates rather than parties.
This difference manifests itself in the definition of a committee because a single candidate cannot receive multiple seats.
That is, a \textit{candidate committee} $W$ is now simply a subset of $\C$ with cardinality $k$. (Therefore, it is usually assumed that $|C|\ge k$.)
A \textit{\ca rule} is a function that maps each \ca election to a candidate committee.

A diverse set of such voting rules has been proposed since the late 19th century~\citep{KiMa12a,Jans16a,LaSk20v3}, out of which we will only introduce the one which we use for our main positive result.
Let $H_j$ denote the $j$th harmonic number, \ie $H_j = \sum_{t=1}^{j} 1/t$.
Given $(\N, \C, \A, k)$, the \ca rule \emph{proportional approval voting (PAV)}, introduced by \citet{Thie95a}, chooses a candidate committee $W$ maximizing the \textit{PAV score} $\PAV(W) = \sum_{i \in \N} H_{|W \cap \A_i|}$.

We now describe EJR and core stability in the \ca setting, from which we derived our versions. Recall that $q(S) = \lfloor k \, |S| / n \rfloor$.
A candidate committee $W$ provides \textit{EJR} if there is no subset $S \subseteq \N$ and no integer $\l > 0$ such that $q(S) \ge \l$, $|\bigcap_{i \in S} \A_i| \geq \l$, and $|\A_i \cap W| < \l$ for all $i \in S$.
(The requirement $|\bigcap_{i \in S} \A_i| \geq \l$ is often called \emph{cohesiveness}.)
A \ca rule satisfies EJR if it always produces EJR committees.

The definition of core stability is even closer to the version in \pa elections:
A candidate committee $W$ is \textit{core stable} if there is no nonempty group $S \subseteq N$ and no set $T \subseteq \C$ of size $|T| \le q(S)$ such that $|A_i \cap T| > |A_i \cap W|$ for all $i \in S$.
The \textit{core} consists of all core-stable candidate committees.

\subsection{Embedding Party-Approval Elections} \label{subsec:embeddingPappCapp}
We have informally argued in \Cref{sec:relations} that \pa elections constitute a subdomain of \ca elections.
We formalize this notion by providing an embedding of \pa elections into the \ca domain.
Our approach is similar to that of \citet{BLS18a}, who have formalized how apportionment problems can be phrased as \ca elections.

For a given \pa election $(\N, \P, \A, k)$, we define a corresponding \ca election $(\N, C, \A', k)$ with the same set of voters~$N$ and the same committee size $k$. The set of candidates contains $k$ many ``clone'' candidates $p^{(1)}, \dots, p^{(k)}$ for each party $p \in P$, so $C = \bigcup_{p\in P} \{p^{(1)}, \dots, p^{(k)}\}$.
Voter $i$ approves a candidate~$p^{(j)}$ in the candidate-approval election if and only if she approves the corresponding party~$p$ in the \pa election. Thus, $A_i' = \bigcup_{p \in A_i} \{p^{(1)}, \dots, p^{(k)}\}$.
This embedding establishes \pa elections as a subdomain of \ca elections. As a consequence, we can apply rules from the more general candidate-approval setting to the \pa setting, by
\begin{enumerate}
	\item translating the \pa election into a \ca election, 
	\item applying the \ca rule, and 
	\item counting the number of chosen clones per party to construct a committee over parties.
\end{enumerate}

\begin{remarkfixed}
\label{remark:unary}
By our assumption that $k$ is encoded in unary for the purpose of complexity analysis (see \Cref{sec:papp}), the translation of a \pa{} election yields a polynomial-sized \ca{} election.
Thus, a polynomial-time \ca{} rule applied to the \pa{} election runs in polynomial time as well.
If $k$ was instead encoded in binary, elections with large $k$ and few parties could be described so concisely that even straightforward candidate-approval algorithms would formally have exponential running time.\footnote{However, some rules may admit implementations that remain efficient for binary $k$. For example, Rule X \citep{PeSk20awitharxiv} will repeatedly assign seats to the same party until one of its supporters runs out of virtual money. Since this happens at most $n$ times, this observation can be used to design an efficient algorithm (with runtime depending on $\log k$ instead of $k$). Still, a linear time dependence on $k$ is acceptable in most applications.}
(The same issue does not appear in \ca{} elections, where we need to list at least $k$ candidates, which makes the description verbose.)
\end{remarkfixed}

\label{axiomsderivation}
Having established \pa elections as a subdomain of \ca elections, our variants of EJR and core stability (Definitions~\ref{def:ejr} and~\ref{def:core}) are immediately induced by their \ca counterparts.
Any \ca rule satisfying an axiom in the \ca setting will satisfy the corresponding axiom in the \pa setting as well.
Note that, by restricting our view to party approval, the cohesiveness requirement of EJR is reduced to requiring a single commonly approved party.

\subsection{PAV Guarantees Core Stability}
\label{subsec:pavCore}
A powerful stability concept in economics, core stability is a natural extension of EJR\@.
It is particularly attractive because blocking coalitions do not need to unanimously approve any party; they only need to be able to coordinate for mutual gain.

Unfortunately, it is still unknown whether core-stable candidate committees exist for all \ca elections.\footnote{However, it is known that \textit{approximately} core-stable committees exist, for several different ways of approximating the core notion \citep{FMS18,CJM+19,JMK20a,PeSk20awitharxiv}.}
All standard candidate-approval rules either already fail weaker representation axioms such as EJR, or are known to fail core stability. In particular, PAV satisfies EJR, but may produce non-core-stable committees for \ca elections \citep{ABC+16a}.
\citet{PeSk20awitharxiv} show that a large class of \ca rules (so-called welfarist rules) must all fail core stability.

For our main result, we show that core stability can always be achieved in the \pa setting. Specifically, the committee selected by PAV is core stable for \pa elections.
Our proof uses a similar technique to the proof that PAV satisfies EJR for \ca elections \citep[Theorem 10]{ABC+16a}; we discuss the essential difference in \Cref{remark:core-disjoint}.
\begin{theorem}
    For every \pa election, PAV chooses a core-stable committee. Hence, the core of a \pa election is nonempty.
    \label{thm:pav_papp_core}
\end{theorem}

\begin{proof}
Consider a \pa election $(\N,\P,\A,k)$ and let $\W_1: P \rightarrow \mathbb{N}$ be the committee selected by PAV\@.
    Assume for a contradiction that $\W_1$ is not core stable. Then there is a nonempty coalition $S\subseteq N$ and a committee $T : \P \to \mathbb{N}$ such that $|T| \leq q(S) \leq  k \, |S| / n$ and $u_i(T) \ge u_i(W_1) + 1$ for every voter $i \in S$.

   For each party $p$, we let $\MC^+(p, \W_1)$ denote the marginal increase of the PAV score when we allocate an extra seat to $p$. Thus,
   \[
   \MC^+(p,\W_1)
   = \PAV(\W_1 + \{p\}) - \PAV(\W_1)
   = \sum_{i \in N_p} \frac{1}{u_i(\W_1) + 1},
   \]
   where $N_p =\{i \in N \mid p \in A_i\}$. 
   Let us calculate the average marginal increase when adding an elements of $T$:
   \begin{align*}
        \frac{1}{|T|} \sum_{p \in P} T(p) \, \MC^+(p, \W_1) 
        &= \frac{1}{|T|} \sum_{i \in \N} \sum_{p \in A_i} \frac{T(p)}{u_i(\W_1) + 1}
        \geq \frac{1}{|T|} \sum_{i \in S} \sum_{p \in A_i} \frac{T(p)}{u_i(\W_1) + 1} \\
        &\geq \frac{1}{|T|} \sum_{i \in S} \sum_{p \in A_i} \frac{T(p)}{u_i(T)} 
        = \frac{1}{|T|} \sum_{i \in S} \frac{u_i(T)}{u_i(T)} = \frac{|S|}{|T|} \geq \frac{n}{k}.
    \end{align*}
Thus, there is a party $p_1$ with $\MC^+(p_1,W_1) \ge n/k$. Let $W_2 = W_1 + \{p_1\}$.

	Next, for each party $p$ with $W_2(p) > 0$, let $\MC^-(p, W_2)$ be the marginal decrease of the PAV score if we take away a seat from $p$ in $W_2$. Thus,
	\[
	   \MC^-(p,W_2)
	   = \PAV(W_2)-\PAV(W_2 - \{p\})
	   = \sum_{i \in N_p} \frac{1}{u_i(W_2)}.
	\]
	The average marginal decrease of taking away a seat from $W_2$ is
	\begin{align*}
		\frac{1}{k+1} \sum_{p\in P} W_2(p) \, \MC^-(p,W_2)
		&= \frac{1}{k+1} \sum_{p\in P} \sum_{i \in N_p} \frac{W_2(p)}{u_i(W_2)} \\
		&= \frac{1}{k+1} \sum_{i \in N} \sum_{p \in A_i} \frac{W_2(p)}{u_i(W_2)} \\
		&= \frac{1}{k+1} |\{ i \in N : u_i(W_2) > 0 \}|
		\leq \frac{n}{k+1}.
	\end{align*}
	Thus, there is some party $p_2$ with $W_2(p_2) > 0$ such that $\MC^-(p_2, W_2) \le \frac{n}{k+1}$. Write $W_3 = W_2 - \{p_2\} = W_1 + \{p_1\} - \{p_2\}$. Then
	\begin{align*}
		\PAV(W_3)
		&= \PAV(W_2) - \MC^-(p_2, W_2) \\
		&= \PAV(\W_1) + \MC^+(p_1,W_1) - \MC^-(p_2, W_2) \\
		&\geq \PAV(\W_1) + \tfrac{n}{k} - \tfrac{n}{k+1} \\
		&> \PAV(\W_1),
	\end{align*}
	contradicting the optimality of $\W_1$.
\end{proof}

\begin{remarkfixed}
	Our proof of \Cref{thm:pav_papp_core} can be easily adapted to show that PAV satisfies the stronger version of core defined with respect to the \emph{Droop quota} \citep{Droo81a,Jans18a}, by assuming $|T| < (k + 1) |S|/n$ rather than $|T| \leq k |S| / n$.
\end{remarkfixed}

\begin{remarkfixed}
	\label{remark:core-disjoint}
	For candidate-approval elections, the proof of \Cref{thm:pav_papp_core} shows that PAV satisfies core stability restricted to ``disjoint objections'': if $W$ is the committee selected by PAV, then there can be no set $T$ with $T \cap W = \emptyset$ such that there is a coalition $S$ with $T \le q(S)$ and $u_i(T) > u_i(W)$ for all $i \in S$. Note that with our embedding of party-approval elections into candidate-approval elections, the disjointness assumption is without loss of generality, and hence PAV satisfies core stability for party-approval elections. The disjoint objections property also implies the result of \citet[Thm.~6]{PeSk20awitharxiv} that PAV satisfies the ``2-core'' property in the candidate-approval context: If there was an objection $T$ that more than doubled the utility of each coalition member, then $T \setminus W$ would be a disjoint core deviation, which is a contradiction.
\end{remarkfixed}

\begin{remarkfixed}
Because $H_j = \Theta(\log{j})$, 
the PAV objective is closely related to the classical \textit{maximum Nash welfare (MNW)} solution \citep{Nash50b,KaNa79a}. One can see PAV as a discretization of the MNW solution for selecting a probability distribution $\sigma : P \to [0,1]$ over parties, where we can interpret $\sigma(p)$ as the fraction of seats that should be allocated to party $p$. That rule satisfies a continuous analog of the core condition \citep{FGM16b,ABM17a}. 
However, other natural discretizations of the Nash rule do not satisfy the core condition. In the next section, we will see that discretizing the Nash rule using common apportionment methods leads to violations of core stability. Furthermore, selecting a committee that maximizes Nash welfare (rather than the PAV objective function) may fail core stability, even in party-choice elections \citep[Theorem~2]{BLS18a}.
\end{remarkfixed}

Given that PAV satisfies core stability in \pa elections but not in \ca elections, do other \ca rules satisfy stronger representation axioms when restricted to the \pa subdomain? We have studied this question for various rules besides PAV, and the answer was always negative; see \cref{sec:app:axiomatic} for details.\footnote{We present relevant counterexamples for the \ca rules \seqP, \maxP, \eneP, Rule X, and the Maximin Support Method. In addition, we verified for the \ca rules SeqPAV, RevSeqPAV, \varP, Approval Voting (AV), SatisfactionAV, MinimaxAV,  MonroeAV, GreedyMonroeAV, GreedyAV, HareAV, and Chamberlin--CourantAV that existing counterexamples can easily be adjusted to the \pa setting.}

A major drawback of PAV is that it fails committee monotonicity, and PAV continues to fail this axiom in the \pa setting.\footnote{Existing counterexamples for the \ca setting \citep{LaSk20v3} can be adapted in a straight-forward way.}
Therefore, parties may lose seats when the committee size is increased. In the next section, we construct \pa rules that avoid this undesirable behavior.

\section{Constructing Party-Approval Rules via Portioning and Apportionment}
\label{sec:using_app}
Party-approval elections are a generalization of party-choice elections, which can be thought of as \pa elections in which all approval sets are singletons.
Since there is a rich body of research on apportionment methods \citep{BaYo82a,Puke14a} which act on party-choice elections, it is natural to examine whether we can employ these methods for our setting as well. To use them, we will need to translate \pa elections into the party-choice domain on which apportionment methods operate. 
This translation thus needs to transform a collection of approval votes over parties into vote shares for each party.
Motivated by time sharing, \citet{BMS05} have developed a theory of such transformation rules, further studied by \citet{Dudd15a} and \citet{ABM19}. We will refer to this framework as \emph{portioning}.

The approach explored in this section, then, divides the construction of a \pa rule into two independent steps: (1)~portioning, which maps a \pa election to a vector of parties' shares; followed by (2)~apportionment, which transforms the shares into a seat distribution.

Both the portioning and the apportionment literature have discussed representation axioms similar in spirit to EJR and core stability. For both settings, several rules have been found to satisfy these properties. One might hope that by composing two rules that are each representative, we obtain a \pa rule that is also representative (and satisfies, say, EJR). If we succeed in finding such a combination, it is likely that the resulting voting rule will automatically satisfy committee monotonicity since most apportionment methods satisfy this property. In the general candidate-approval setting (considered in \Cref{sec:using_capp}), the existence of a rule satisfying both EJR and committee monotonicity is an open problem.

\subsection{Preliminaries}

We start by introducing relevant notions from the literature on portioning~\citep{BMS05,ABM19} and apportionment \citep{BaYo82a,Puke14a}, with notation suitably adjusted to our setting. 

\subsubsection{Portioning}
A \textit{portioning problem} is a triple $(\N, \P, \A)$, just as in \pa voting but without a committee size.
A \emph{portioning} is a function $r: P \rightarrow [0,1]$ with $\sum_{p \in P} r(p)=1$. We interpret $r(p)$ as the vote share of party $p$. 
A \emph{portioning method} maps each portioning problem $(\N, \P, \A)$ to a portioning.

Our minimum requirement on portioning methods will be that they uphold proportionality if all approval sets are singletons, \ie if we are already in the party-choice domain.
Formally, we say that a portioning method is \emph{faithful} if for all $(\N, \P, \A)$ with $|A_i|=1$ for all $i \in N$, the resulting portioning $r$ satisfies $r(p) = |\{i \in \N \mid \A_i = \{p\}\}| / n$ for all $p \in P$.
Among the portioning methods considered by \citet{ABM19},
only three are faithful. They are defined as follows. 
\begin{description}
\item[Conditional utilitarian portioning] selects, for each voter $i$, $p_i$ as a party in $A_i$ approved by the highest number of voters. Then, $r(p) = |\{i \in N \mid p_i = p\}| / n$ for all $p \in P$.

\item[Random priority] computes $n!$ portionings, one for each permutation $\sigma$ of $N$, and returns their average.
The portioning for $\sigma = (i_1, \dots, i_n)$  maximizes $\sum_{p \in A_{i_1}}\! r(p)$, breaking ties by maximizing $\sum_{p \in A_{i_2}}\! r(p)$, and so forth.

\item[Nash portioning] selects the portioning $r$ maximizing the Nash welfare $\prod_{i \in \N} \big( \sum_{p \in \A_i} \! r(p) \big)$. 
\end{description}
\noindent
When computing the outcomes of these rules, ties may occur. For our results it will not matter how ties are broken: we only use these rules in counterexamples in which no ties occur.

On first sight, Nash portioning seems particularly promising because it satisfies portioning versions of core stability and EJR \citep{ABM19,GuNe14a}. Concretely, it satisfies a property called \textit{average fair share} introduced by \citet{ABM19}, which requires that there is no subset $S\subseteq N$ of voters such that $\bigcap_{i \in S} A_i \neq \emptyset$ and $\frac{1}{|S|}\sum_{i\in S}\sum_{p \in A_i} r(p) < |S|/|N|$. However, despite these promising properties, we will see that Nash portioning does not work for our purposes. Instead, we will need to make use of a more recent portioning approach, which was proposed by \citet{SpGe19a} in the context of \pa voting.
\begin{description}
\item[Majoritarian portioning]
proceeds in rounds $j = 1, 2, \dots$.
Initially, all parties and voters are \emph{active}.
In iteration $j$, select the active party $p_j$ that is approved by the highest number of active voters.
Let $N_j$ be the set of active voters who approve $p_j$.
Then, set $r(p_j)$ to $|N_j| / n$, and mark $p_j$ and all voters in $N_j$ as inactive.
If active voters remain, start the next iteration; otherwise, return $r$.
\end{description}
\noindent
Under majoritarian portioning, we ignore the approval preferences of voters after they have been ``assigned'' to a party. 
Note that conditional utilitarian portioning is a similar sequential method which does, however, not ignore the preferences of inactive voters.

\subsubsection{Apportionment}
An \emph{apportionment problem} is a tuple $(\P, r, k)$, which consists of a finite set of parties $\P$, a portioning $r : \P \to [0, 1]$ specifying the vote shares of parties, and a committee size $k \in \mathbb{N}$.
Committees are defined as for \pa elections, and an \emph{apportionment method} maps apportionment problems to committees~$\W$ of size~$k$.

An apportionment method satisfies \emph{lower quota} if each party $p$ is always allocated at least $\lfloor k \cdot r(p) \rfloor$ seats in the committee.
Furthermore, an apportionment method $f$ is \emph{committee monotonic} if  $f(\P, r, k) \subseteq f(\P, r, k + 1)$ for every apportionment problem $(\P, r, k)$.

Among the standard apportionment methods, only two satisfy both lower quota and committee monotonicity: the \textit{D'Hondt method} (aka \textit{Jefferson method}) and the \emph{quota method}.\footnote{All other \textit{divisor methods} fail lower quota, and the \textit{Hamilton method} is not committee monotonic \citep{BaYo82a}.}
The D'Hondt method assigns the $k$ seats iteratively, each time giving the next seat to the party $p$ with the largest quotient $r(p)/(s(p)+1)$, where $s(p)$ denotes the number of seats already assigned to $p$.
The \emph{quota method} \citep{BaYo75a} is identical to the D'Hondt method, except that, in the $j$th iteration, only parties $p$ satisfying $s(p)/j<r(p)$ are eligible for the allocation of the next seat. Ties may be broken arbitrarily.

\subsubsection{Composition}
If we take any portioning method and any apportionment method, we can compose them to obtain a \pa rule. 
Formally, the composition of portioning method $R$ and apportionment method $M$ maps each \pa election $(N,P,A,k)$ to a committee $M(P,R(N,P,A),k)$.
Note that if the apportionment method is committee monotonic then so is the composed rule, since the portioning is independent of $k$.

\subsection{Composed Rules That Fail EJR}
Perhaps surprisingly, many pairs of portioning and apportionment methods fail EJR\@.
This is certainly true if the individual parts are not representative themselves. For example, if an apportionment method $M$ ``properly'' fails lower quota (in the sense that there is a rational-valued input $r$ on which lower quota is violated), then one can construct an example profile on which any composed rule using $M$ fails EJR:  Construct a \pa election with singleton approval sets in which the voter counts are proportional to the shares in the counterexample $r$. Then any faithful portioning method, applied to this election, must return $r$. Since $M$ fails lower quota on $r$, the resulting committee will violate EJR\@.
By a similar argument, suppose that an apportionment method violates committee monotonicity, and that there is a rational-valued counterexample. Then the apportionment method, when composed with a faithful portioning method, will give rise to a \pa rule that fails committee monotonicity.

As mentioned above, D'Hondt and the quota method are the only standard apportionment rules to satisfy both lower quota and committee monotonicity.
However, the composition of either option with the conditional utilitarian, random priority, or Nash portioning methods fails EJR, as the following examples show.

\begin{example}
    Let $n = k = 6$, $\P = \{p_0, p_1, p_2, p_3\}$, 
    and consider the ballot profile
    $\A = (\{p_0\}, \{p_0\}, \{p_0, p_1, p_2\}, \{p_0, p_1, p_2\}, \{p_1, p_3\}, \{p_2, p_3\})$.

    Then, the conditional utilitarian solution sets $r(p_0) = 4/6$, $r(p_1) = r(p_2) = 1/6$, and $r(p_3) = 0$.
    Any apportionment method satisfying lower quota allocates four seats to $p_0$, one each to $p_1$ and $p_2$, and none to $p_3$.
    The resulting committee does not provide EJR since the last two voters, who jointly approve~$p_3$, have a quota of $q(\{5,6\})=2$ that is not met. \qed
\end{example}

\begin{example}
    \label{ex:nash-fails-ejr}
    Let $n = k = 6$, $\P = \{p_0, p_1, p_2, p_3\}$, 
    and consider the ballot profile 
    $\A = (\{p_0\}, \{p_0\}, \{p_0, p_1, p_2\}, \{p_0, p_1, p_3\}, \{p_1\}, \{p_2, p_3\})$.

    Random priority chooses the portioning $r(p_0) = 23/45$, $r(p_1) = 23/90$, and $r(p_2) = r(p_3) = 7/60$. 
    Both D'Hondt and the quota method allocate four seats to $p_0$, two seats to $p_1$, and none to the other two parties.
    This violates the claim to representation of the sixth voter (with $q(\{6\})=1$).

    Nash portioning produces a fairly similar portioning, with $r(p_0) \approx 0.5302$, $r(p_1) \approx 0.2651$, and $r(p_2) = r(p_3) \approx 0.1023$.
    D'Hondt and the quota method produce the same committee as above, leading to the same EJR violation. \qed
\end{example}
It might be surprising that Nash portioning combined with a lower-quota apportionment method violates EJR. After all, Nash portioning satisfies core stability in the portioning setting, 
which is a strong notion of proportionality, and the lower-quota property limits the rounding losses when moving from the portioning to a committee.
As expected, in the election of Example~\ref{ex:nash-fails-ejr}, the portioning produced by Nash gives sufficient representation to the sixth voter since $r(p_2) + r(p_3) \approx 0.2047 > 1/6$.
However, since both $r(p_2)$ and $r(p_3)$ are below $1 / 6$ on their own, lower quota does not apply to either of the two parties, and the sixth voter loses all representation in the apportionment step.\footnote{There are similar examples where Nash portioning with D'Hondt apportionment violates EJR even though every party receives at least one seat, and examples where EJR is violated by a margin of more than one seat.}

\subsection{Composed Rules That Satisfy EJR}
\label{sec:composed-italian}

As we have seen, several initially promising portioning methods fail to compose to a rule that satisfies EJR\@. 
One reason is that these portioning methods are happy to assign small shares to several parties. The apportionment method may round several of those small shares down to zero seats. This can lead to a failure of EJR when not enough parties obtain a seat. It is difficult for an apportionment method to avoid this behavior since the portioning step obscures the relationships between different parties that are apparent from the approval ballots of the voters.

Since majoritarian portioning maximizes the seat allocations to the largest parties, it tends to avoid the problem we have just identified. While it fails the strong representation axioms that Nash portioning satisfies, this turns out not to be crucial: Composing majoritarian portioning with any apportionment method satisfying lower quota yields an EJR rule. If we use an apportionment method that is also committee monotonic, such as D'Hondt or the quota method, we obtain a \pa rule that satisfies both EJR and committee monotonicity.

\begin{theorem}\label{thm:italian}
Let $M$ be a committee monotonic apportionment method satisfying lower quota. Then, the \pa rule composing majoritarian portioning and $M$
satisfies EJR and committee monotonicity. 
\end{theorem}
\begin{proof}
    Consider a \pa election $(N,P,A,k)$ and let $r$ be the outcome of  majoritarian portioning applied to $(N,P,A)$.
Let $N_1, N_2, \dots$ and $p_1, p_2, \dots$ be the voter groups and parties in the construction of majoritarian portioning, so that 
	$r(p_j) = |N_j|/n$ for all $j$.
	
	Consider the committee $W = M(P,r,k)$ and suppose that EJR is violated, i.e., that there exists a group $S \subseteq N$ with $\bigcap_{i \in S} A_i \neq \emptyset$ and $u_i(W) < q(S)$ for all $i \in S$. 
	
	Let $j$ be minimal such that $S \cap N_j \neq \emptyset$. 
	We now show that $|S| \leq |N_j|$.
    By the definition of $j$, no voter in $S$ approves of any of the parties $p_1, p_2, \dots p_{j-1}$; thus, all those voters remain active in round $j$.
    Consider a party $p^* \in \bigcap_{i \in S} A_i$.
    In the $j$th iteration of majoritarian portioning, this party had an approval score (among active voters) of at least $|S|$.
    Therefore, the party $p_j$ chosen in the $j$th iteration has an approval score that is at least $|S|$ (of course, $p^* = p_j$ is possible).
    The approval score of party $p_j$ equals $|N_j|$. Therefore, $|N_j| \ge |S|$. 
	
	Since $|N_j| \geq |S|$, we have $q(N_j)\geq q(S)$. 
	Since $M$ satisfies lower quota, it assigns at least 
	$\lfloor k \cdot r(p_j) \rfloor = \lfloor k \, (|N_j|/n) \rfloor = q(N_j)$ 
	seats to party $p_j$. 
    Now consider a voter $i \in S \cap N_j$.
    Since this voter approves party $p_j$, we have $u_i(W) \ge W(p_j) \geq q(N_j) \ge q(S)$, a contradiction.

    This shows that EJR is indeed satisfied; committee monotonicity follows from the committee monotonicity of $M$.  
\end{proof}

While the \pa rules identified by \Cref{thm:italian} satisfy EJR and committee monotonicity, they do not reach our gold standard of representation, i.e., core stability:
\begin{example}
\label{ex:italiannocore}
    Let $n = k = 16$, $P = \{p_0, \dots, p_4\}$, and consider the following ballot profile:
    \begin{align*}
        & 4 \times \{p_0, p_1\}, \qquad 3 \times \{p_1, p_2\}, \qquad 1 \times \{p_2\} \\
        & 4 \times \{p_0, p_3\}, \qquad 3 \times \{p_3, p_4\}, \qquad 1 \times \{p_4\}
    \end{align*}
Majoritarian portioning allocates $1/2$ to~$p_0$ and $1/4$ each to~$p_2$ and~$p_4$.
    Any lower-quota apportionment method must translate this into 8 seats for $p_0$ and 4 seats each for $p_2$ and $p_4$.
    This committee is not in the core:
    Let $S$ be the coalition of all 14 voters who approve multiple parties, and
    let $T$ allocate 4 seats to $p_0$ and 5 seats each to $p_1$ and $p_3$.
    This gives strictly higher representation to all members of the coalition. \qed
\end{example}

The example makes it obvious why majoritarian portioning cannot satisfy core stability:
All voters approving of $p_0$ get deactivated after the first round, which makes $p_2$ seem universally preferable to $p_1$.
However, $p_1$ is a useful vehicle for cooperation between the group approving $\{p_0, p_1\}$ and the group approving $\{p_1, p_2\}$.
Since majoritarian portioning is blind to this opportunity, it cannot guarantee core stability.

The example also illustrates the power of core stability:
The deviating coalition does not agree on any single party they support, but would nonetheless benefit from the deviation. Core stability is sensitive to this demand for better representation.

\section{Computational Aspects}

To use a voting rule, we need to compute its output. Ideally, we would like an efficient (i.e., polynomial-time) algorithm for this task, so that we can announce the voting outcome soon after all votes have been cast. Fortunately, many rules admit fast algorithms. For example, the composed rules from Section~\ref{sec:composed-italian} can be computed efficiently as long as the employed apportionment method is computable in polynomial time (which is the case for D'Hondt and the quota method). In addition, by our discussion in Remark~\ref{remark:unary}, every multiwinner voting rule that runs in polynomial time for the \ca setting also runs in polynomial time for the \pa setting.

That being said, given our result about core stability in 
Section \ref{sec:using_capp}, we are particularly interested in computing the outcome of PAV, which is NP-hard to compute in the \ca setting \citep{AGG+15a}. Since \pa elections are a restricted domain, it is in principle possible that PAV is easier to compute on that domain, but, as we show in \cref{app:compAspects}, hardness still holds for \pa elections.

\begin{restatable}{theorem}{restatePAVhard}
    For a given party-approval election and threshold $s \in \mathbb{R}$, deciding whether there exists a committee with PAV score at least $s$ is NP-hard.\label{thm:pav_np_hard}
\end{restatable}

Equally confronted with the computational complexity of PAV, \citet{AEH+18a} proposed a local-search variant of PAV, which runs in polynomial time and guarantees EJR in the \ca setting. Using the same approach, we can find a core-stable committee in the \pa setting.

\begin{restatable}{theorem}{restateCoreInPoly}
    Given a \pa election, a core-stable committee can be computed in polynomial time. 
    \label{thm:pav_comp_core}
\end{restatable}
We defer the proof of this theorem to \cref{app:compAspects}.
In \cref{app:phragrules}, we additionally show that an optimization variant of \phrag's rule \citep{BFJL16a} remains intractable in the \pa subdomain.

\smallskip

\citet{LaSk20v3} posed as an open problem to determine the complexity of checking whether a given committee satisfies core stability. We show that the problem is coNP-complete. Our proof is written for \pa elections, but the result implies hardness for the candidate-approval setting because \pa elections are a special case of \ca elections.

\begin{theorem}
    For a given party-approval (or candidate-approval) election and a committee $W$, it is coNP-complete to decide whether $W$ satisfies core stability. 
\end{theorem}

\begin{proof}
    The complement problem is clearly in NP since a core deviation provides a certificate.
We reduce from the NP-complete problem exact cover by $3$-sets (X3C).
    Here, given a set $X$ with $|X| = 3r$ and a collection $\mathcal{B}$ of $3$-element subsets of~$X$, the question is whether there exists a selection $\mathcal{B}' \subseteq \mathcal{B}$ of $r$ of the subsets such that every element of $X$ occurs in one of the sets in $\mathcal{B}'$. 
    
    We construct an instance of our problem as follows: For every set $B \in \mathcal{B}$ we introduce a \emph{set candidate} and for every element in $X$ we introduce an \emph{element voter}. We set $k=r$ and introduce one \emph{special voter}, $k-1$ \emph{private candidates} and one \emph{dummy candidate}. The approval sets are as follows: Each element voter $x \in X$ approves exactly those set candidates $B \in \mathcal{B}$ with $x \in B$ and the special voter approves all candidates except the dummy candidate. (Thus, no voter approves the dummy candidate.) Finally, let $W$ be the committee consisting of the private candidates and the dummy candidate.\footnote{The committee $W$ assigns seats only to Pareto-dominated parties, making it clearly suboptimal. One can adjust the reduction to show that the problem remains hard for committees~$W$ that do not give seats to Pareto-dominated parties.}
    Note that $|W| = k$. We claim that $W$ is not core stable if and only if the X3C instance is a \emph{yes} instance.

    Suppose that $W$ is not core stable, and let $S \subseteq N$ and committee $T : P \to \mathbb N$ witness this fact. Without loss of generality, we may assume that $T$ only gives seats to set candidates, since all other candidates are dominated by set candidates. Suppose $|T| = t$. Then $T$ provides positive utility to at most $3t$ element voters. These $3t$ voters on their own can afford $\lfloor 3t \cdot k / n \rfloor \leq 3t \cdot k / (3k + 1) < t$  candidates. Because all element voters in $S$ must obtain positive utility from $T$, it follows that the special voter must be part of $S$. Because the special voter $i$ has $u_i(T) > u_i(W) = k - 1$, we have $u_i(T) = k$. Thus $|T| = k$, and a committee of this size can only be afforded by the grand coalition, so $S = N$. Thus, every element voter is part of $S$ and thus obtains positive utility from $T$, and hence for every element, $T$ contains at least one set candidate corresponding to a set containing that element. It follows that the X3C instance has a solution.
    
    Conversely, every solution to the X3C instance induces a committee $T$ consisting of the $k$ set candidates chosen by the solution. Then $T$ gives positive utility to all element voters and increases the special voter's utility from $k-1$ to $k$. Hence $T$ together with $S = N$ show that $W$ is not core stable.
\end{proof}

In the candidate-approval setting, checking whether a given committee satisfies EJR is coNP-complete~\citep{ABC+16a,AEH+18a}. In other words, given a committee, it is hard to find a cohesive coalition of voters that is underrepresented. Interestingly, this task is tractable in \pa elections. Intuitively, checking becomes easier in \pa elections as groups of voters are already cohesive when they have only a single approved party in common. 

\begin{theorem}
    Given a \pa election $(\N,\P,\A,k)$ and a committee $\W : \P \to \mathbb{N}$, it can be checked in polynomial time whether $\W$ satisfies EJR\@.
    \label{thm:party_ejr_check}
\end{theorem}
\begin{proof}
We describe a procedure to check whether a given committee $\W$ violates EJR. For each party $p \in P$ and each $\ell \in [k]$, define \[S_{p,\ell} = \{i \in N \mid p \in A_i \text{ and } u_i(W) \leq \ell\}\] 
    and check whether $\ell < q(S_{p,\ell})$ holds. If so, the set $S_{p,\ell}$ induces an EJR violation. This is because $\bigcap_{i \in S_{p,\ell}}A_i \neq \emptyset$ and $ u_i(W) \leq \ell < q(S_{p},\ell)$ holds for all $i \in S_{p,\ell}$. 
    
    Now, assume that the condition is not satisfied for any party $p \in P$ and any $\ell \in [k]$. We claim that this proves the nonexistence of an EJR violation. Assume for contradiction that there exists a group $S \subseteq N$ inducing an EJR violation. Let $p \in \bigcap_{i \in S} A_i$ and $\ell = \max_{i \in S} u_i(W)$. By definition, $S \subseteq S_{p,\ell}$ and hence $q(S_{p,\ell})\geq q(S)>\ell$, a contradiction. 
A straightforward implementation of this algorithm has polynomial running time $\mathcal{O}( |\P| \, k \, n)$.  
\end{proof}

We observe a similar effect for \emph{proportional justified representation (PJR)}, a proportionality axiom introduced by \cite{SFF+17a} which is weaker than EJR. While checking whether a committee satisfies PJR is coNP-complete in the candidate-approval setting, we can solve the problem in polynomial time via submodular minimization in our setting. For a formal definition and the proof, see Appendix \ref{sec:checkPJR}.

\section{Discussion}
In this paper, we have initiated the axiomatic analysis of approval-based apportionment.
On a technical level, it would be interesting to see whether the \pa domain allows us to satisfy other combinations of axioms that are not known to be attainable in \ca elections.
For instance, the compatibility between strong representation axioms and certain notions of support monotonicity is an open problem~\citep{SF19}.

We have presented our setting guided by the application of apportioning parliamentary seats to political parties.
But our formal setting has other interesting applications.
An example would be participatory budgeting settings where items all have equal costs and come in different types.
For instance, a university department could decide how to allocate Ph.D.\ scholarships across different research projects, in a way that respects the preferences of funding organizations.

As another example, the literature on multiwinner elections suggests many applications to recommendation problems~\citep{SFL16a}.
For instance, one might want to display a limited number of news articles, movies, or advertisements in a way that fairly represents the preferences of the audience.
These preferences might be expressed not over individual pieces of content, but over content producers (such as newspapers, studios, or advertising companies), in which case our setting provides rules that decide how many items should be contributed by each source.
Expressing preferences on the level of content producers is natural in repeated settings, where the relevant pieces of content change too frequently to elicit voter preferences on each occasion.
Besides, content producers might reserve the right to choose which of their content should be displayed.

In the general \ca setting, the search continues for rules that satisfy EJR and committee monotonicity, or core stability.
But for the applications mentioned above, these guarantees are already achievable today.

\paragraph{Acknowledgements}
This work was partially supported by the Deutsche Forschungsgemeinschaft under grant BR~4744/2-1. 
We thank Steven Brams and Piotr Skowron for suggesting the setting of party approval to us, and we thank 
Rupert Freeman,
Levi Geiser,
Anne-Marie George,
Ayumi Igarashi, Svante Janson,
Jérôme Lang,
Ariel Procaccia, and the anonymous reviewers
for helpful comments and discussions.

\bibliographystyle{plainnat}

\begin{thebibliography}{55}
	\providecommand{\natexlab}[1]{#1}
	\providecommand{\url}[1]{\texttt{#1}}
	\expandafter\ifx\csname urlstyle\endcsname\relax
	\providecommand{\doi}[1]{doi: #1}\else
	\providecommand{\doi}{doi: \begingroup \urlstyle{rm}\Url}\fi
	
	\bibitem[Al\'os-Ferrer(2006)]{Alos06a}
	C.~Al\'os-Ferrer.
	\newblock A simple characterization of approval voting.
	\newblock \emph{Social Choice and Welfare}, 27\penalty0 (3):\penalty0 621--625,
	2006.
	
	\bibitem[Aziz et~al.(2015)Aziz, Gaspers, Gudmundsson, Mackenzie, Mattei, and
	Walsh]{AGG+15a}
	H.~Aziz, S.~Gaspers, J.~Gudmundsson, S.~Mackenzie, N.~Mattei, and T.~Walsh.
	\newblock Computational aspects of multi-winner approval voting.
	\newblock In \emph{Proceedings of the 14th International Conference on
		Autonomous Agents and Multiagent Systems (AAMAS)}, pages 107--115. IFAAMAS,
	2015.
	
	\bibitem[Aziz et~al.(2017)Aziz, Brill, Conitzer, Elkind, Freeman, and
	Walsh]{ABC+16a}
	H.~Aziz, M.~Brill, V.~Conitzer, E.~Elkind, R.~Freeman, and T.~Walsh.
	\newblock Justified representation in approval-based committee voting.
	\newblock \emph{Social Choice and Welfare}, 48\penalty0 (2):\penalty0 461--485,
	2017.
	
	\bibitem[Aziz et~al.(2018)Aziz, Elkind, Huang, Lackner,
	S{\'a}nchez-Fern{\'a}ndez, and Skowron]{AEH+18a}
	H.~Aziz, E.~Elkind, S.~Huang, M.~Lackner, L.~S{\'a}nchez-Fern{\'a}ndez, and
	P.~Skowron.
	\newblock On the complexity of extended and proportional justified
	representation.
	\newblock In \emph{Proceedings of the 32nd AAAI Conference on Artificial
		Intelligence (AAAI)}, pages 902--909. AAAI Press, 2018.
	
	\bibitem[Aziz et~al.(2019{\natexlab{a}})Aziz, Bogomolnaia, and Moulin]{ABM17a}
	H.~Aziz, A.~Bogomolnaia, and H.~Moulin.
	\newblock Fair mixing: the case of dichotomous preferences.
	\newblock In \emph{Proceedings of the 19th ACM Conference on Economics and
		Computation (ACM-EC)}, pages 753--781, 2019{\natexlab{a}}.
	
	\bibitem[Aziz et~al.(2019{\natexlab{b}})Aziz, Bogomolnaia, and Moulin]{ABM19}
	H.~Aziz, A.~Bogomolnaia, and H.~Moulin.
	\newblock Fair mixing: The case of dichotomous preferences.
	\newblock In \emph{Proceedings of the 2019 {{ACM Conference}} on {{Economics}}
		and {{Computation} (EC)}}, pages 753--781, 2019{\natexlab{b}}.
	
	\bibitem[Balinski and Young(1975)]{BaYo75a}
	M.~L. Balinski and H.~P. Young.
	\newblock The quota method of apportionment.
	\newblock \emph{The American Mathematical Monthly}, 82\penalty0 (7):\penalty0
	701--730, 1975.
	
	\bibitem[Balinski and Young(1982)]{BaYo82a}
	M.~L. Balinski and H.~P. Young.
	\newblock \emph{Fair Representation: {M}eeting the Ideal of One Man, One Vote}.
	\newblock Yale University Press, 1982.
	
	\bibitem[Barber{\`a} and Coelho(2008)]{BaCo08a}
	S.~Barber{\`a} and D.~Coelho.
	\newblock How to choose a non-controversial list with k names.
	\newblock \emph{Social Choice and Welfare}, 31\penalty0 (1):\penalty0 79--96,
	2008.
	
	\bibitem[Bogomolnaia et~al.(2005)Bogomolnaia, Moulin, and Stong]{BMS05}
	A.~Bogomolnaia, H.~Moulin, and R.~Stong.
	\newblock Collective choice under dichotomous preferences.
	\newblock \emph{Journal of Economic Theory}, 122\penalty0 (2):\penalty0
	165--184, 2005.
	
	\bibitem[Brams and Fishburn(2007)]{BrFi07c}
	S.~J. Brams and P.~C. Fishburn.
	\newblock \emph{Approval Voting}.
	\newblock Springer-Verlag, 2nd edition, 2007.
	
	\bibitem[Brams and Kilgour(2014)]{BrKi14a}
	S.~J. Brams and D.~M. Kilgour.
	\newblock Satisfaction approval voting.
	\newblock In \emph{Voting Power and Procedures}, Studies in Choice and Welfare,
	pages 323--346. Springer, 2014.
	
	\bibitem[Brams et~al.(2007)Brams, Kilgour, and Sanver]{BKS07a}
	S.~J. Brams, D.~M. Kilgour, and M.~R. Sanver.
	\newblock A minimax procedure for electing committees.
	\newblock \emph{Public Choice}, 132:\penalty0 401--420, 2007.
	
	\bibitem[Brams et~al.(2019)Brams, Kilgour, and Potthoff]{BKP19a}
	S.~J. Brams, D.~M. Kilgour, and R.~F. Potthoff.
	\newblock Multiwinner approval voting: an apportionment approach.
	\newblock \emph{Public Choice}, 178\penalty0 (1--2):\penalty0 67--93, 2019.
	
	\bibitem[Brill(2021)]{Bril21a}
	M.~Brill.
	\newblock From computational social choice to digital democracy.
	\newblock In \emph{Proceedings of the 30th International Joint Conference on
		Artificial Intelligence (IJCAI) Early Career Spotlight Track}, pages
	4937--4939. IJCAI, 2021.
	
	\bibitem[Brill et~al.(2017)Brill, Freeman, Janson, and Lackner]{BFJL16a}
	M.~Brill, R.~Freeman, S.~Janson, and M.~Lackner.
	\newblock Phragm\'{e}n's voting methods and justified representation.
	\newblock In \emph{Proceedings of the 31st AAAI Conference on Artificial
		Intelligence (AAAI)}, pages 406--413. AAAI Press, 2017.
	
	\bibitem[Brill et~al.(2018)Brill, Laslier, and Skowron]{BLS18a}
	M.~Brill, J.-F. Laslier, and P.~Skowron.
	\newblock Multiwinner approval rules as apportionment methods.
	\newblock \emph{Journal of Theoretical Politics}, 30\penalty0 (3):\penalty0
	358--382, 2018.
	
	\bibitem[Brill et~al.(2020)Brill, G{\"o}lz, Peters, Schmidt-Kraepelin, and
	Wilker]{BGP+19a}
	M.~Brill, P.~G{\"o}lz, D.~Peters, U.~Schmidt-Kraepelin, and K.~Wilker.
	\newblock Approval-based apportionment.
	\newblock In \emph{Proceedings of the 34th AAAI Conference on Artificial
		Intelligence (AAAI)}, pages 1854--1861. AAAI Press, 2020.
	
	\bibitem[Camps et~al.(2019)Camps, Mora, and Saumell]{CMS19a}
	R.~Camps, X.~Mora, and L.~Saumell.
	\newblock The method of {E}nestr{\"o}m and {P}hragm{\'e}n for parliamentary
	elections by means of approval voting.
	\newblock Technical report, arXiv:1907.10590 [econ.TH], 2019.
	
	\bibitem[Cheng et~al.(2019)Cheng, Jiang, Munagala, and Wang]{CJM+19}
	Y.~Cheng, Z.~Jiang, K.~Munagala, and K.~Wang.
	\newblock Group {{Fairness}} in {{Committee Selection}}.
	\newblock In \emph{Proceedings of the 2019 {{ACM Conference}} on {{Economics}}
		and {{Computation}} (EC)}, pages 263--279, 2019.
	
	\bibitem[Droop(1881)]{Droo81a}
	H.~R. Droop.
	\newblock On methods of electing representatives.
	\newblock \emph{Journal of the Statistical Society of London}, 44\penalty0
	(2):\penalty0 141--202, 1881.
	
	\bibitem[Duddy(2015)]{Dudd15a}
	C.~Duddy.
	\newblock Fair sharing under dichotomous preferences.
	\newblock \emph{Mathematical Social Sciences}, 73:\penalty0 1--5, 2015.
	
	\bibitem[Elkind et~al.(2017)Elkind, Faliszewski, Skowron, and Slinko]{EFSS17a}
	E.~Elkind, P.~Faliszewski, P.~Skowron, and A.~Slinko.
	\newblock Properties of multiwinner voting rules.
	\newblock \emph{Social Choice and Welfare}, 48:\penalty0 599--632, 2017.
	
	\bibitem[Fain et~al.(2016)Fain, Goel, and Munagala]{FGM16b}
	B.~Fain, A.~Goel, and K.~Munagala.
	\newblock The core of the participatory budgeting problem.
	\newblock In \emph{Proceedings of the 12th International Workshop on Internet
		and Network Economics (WINE)}, Lecture Notes in Computer Science (LNCS),
	pages 384--399, 2016.
	
	\bibitem[Fain et~al.(2018)Fain, Munagala, and Shah]{FMS18}
	B.~Fain, K.~Munagala, and N.~Shah.
	\newblock Fair allocation of indivisible public goods.
	\newblock In \emph{Proceedings of the 2018 {{ACM Conference}} on {{Economics}}
		and {{Computation}} (EC)}, pages 575--592, 2018.
	
	\bibitem[Fishburn(1978)]{Fish78d}
	P.~C. Fishburn.
	\newblock Axioms for approval voting: {D}irect proof.
	\newblock \emph{Journal of Economic Theory}, 19\penalty0 (1):\penalty0
	180--185, 1978.
	
	\bibitem[Fishburn(1979)]{Fish79a}
	P.~C. Fishburn.
	\newblock Symmetric and consistent aggregation with dichotomous voting.
	\newblock In J.~J. Laffont, editor, \emph{Aggregation and Revelation of
		Preferences}. North-Holland, 1979.
	
	\bibitem[Garey and Johnson(1979)]{GaJo79a}
	M.~R. Garey and D.~S. Johnson.
	\newblock \emph{Computers and Intractability: A Guide to the Theory of
		NP-Completeness}.
	\newblock W. H. Freeman, 1979.
	
	\bibitem[Guerdjikova and Nehring(2014)]{GuNe14a}
	A.~Guerdjikova and K.~Nehring.
	\newblock Weighing experts, weighing sources: {T}he diversity value.
	\newblock Mimeo, 2014.
	
	\bibitem[Israel and Brill(2021)]{IsBr21b}
	J.~Israel and M.~Brill.
	\newblock Dynamic proportional rankings.
	\newblock In \emph{Proceedings of the 30th International Joint Conference on
		Artificial Intelligence (IJCAI)}, pages 261--267. IJCAI, 2021.
	
	\bibitem[Janson(2016)]{Jans16a}
	S.~Janson.
	\newblock Phragm{\'e}n's and {T}hiele's election methods.
	\newblock Technical report, arXiv:1611.08826 [math.HO], 2016.
	
	\bibitem[Janson(2018)]{Jans18a}
	S.~Janson.
	\newblock Thresholds quantifying proportionality criteria for election methods.
	\newblock Technical report, arXiv:1810.06377 [cs.GT], 2018.
	
	\bibitem[Janson and {\"O}berg(2019)]{JO19}
	S.~Janson and A.~{\"O}berg.
	\newblock A piecewise contractive dynamical system and {P}hragm{\'e}n's
	election method.
	\newblock \emph{Bulletin de la Soci{\'e}t{\'e} Math{\'e}matique de France},
	147\penalty0 (3):\penalty0 395--441, 2019.
	
	\bibitem[Jiang et~al.(2020)Jiang, Munagala, and Wang]{JMK20a}
	Z.~Jiang, K.~Munagala, and K.~Wang.
	\newblock Approximately stable committee selection.
	\newblock In \emph{Proceedings of the 52nd Annual ACM SIGACT Symposium on
		Theory of Computing (STOC)}, pages 463--472, 2020.
	
	\bibitem[Kaneko and Nakamura(1979)]{KaNa79a}
	M.~Kaneko and K.~Nakamura.
	\newblock The nash social welfare function.
	\newblock \emph{Econometrica}, 47\penalty0 (2):\penalty0 423--435, 1979.
	
	\bibitem[Kilgour and Marshall(2012)]{KiMa12a}
	D.~M. Kilgour and E.~Marshall.
	\newblock Approval balloting for fixed-size committees.
	\newblock In \emph{Electoral Systems}, Studies in Choice and Welfare, pages
	305--326. Springer, 2012.
	
	\bibitem[Korte and Vygen(2018)]{KV18}
	B.~Korte and J.~Vygen.
	\newblock \emph{Combinatorial Optimization}.
	\newblock {Springer}, 6th edition, 2018.
	
	\bibitem[Lackner and Skowron(2021)]{LaSk20v3}
	M.~Lackner and P.~Skowron.
	\newblock Multi-winner voting with approval preferences.
	\newblock Technical report, arXiv:2007.01795v3 [cs.GT], 2021.
	
	\bibitem[Laslier and Sanver(2010)]{LaSa10a}
	J.-F. Laslier and M.~R. Sanver, editors.
	\newblock \emph{Handbook on Approval Voting}.
	\newblock Studies in Choice and Welfare. Springer-Verlag, 2010.
	
	\bibitem[Mora and Oliver(2015)]{MoOl15a}
	X.~Mora and M.~Oliver.
	\newblock Eleccions mitjan{\c c}ant el vot d'aprovaci{\'o}. {E}l m{\`e}tode de
	{P}hragm{\'e}n i algunes variants.
	\newblock \emph{Butllet{\'\i} de la Societat Catalana de Matem{\`a}tiques},
	30\penalty0 (1):\penalty0 57--101, 2015.
	
	\bibitem[Nash(1950)]{Nash50b}
	J.~F. Nash.
	\newblock The bargaining problem.
	\newblock \emph{Econometrica}, 18\penalty0 (2):\penalty0 155--162, 1950.
	
	\bibitem[Peters and Skowron(2020)]{PeSk20awitharxiv}
	D.~Peters and P.~Skowron.
	\newblock Proportionality and the limits of welfarism.
	\newblock In \emph{Proceedings of the 21st ACM Conference on Economics and
		Computation (ACM-EC)}, pages 793--794. ACM Press, 2020.
	\newblock Full version arXiv:1911.11747 [cs.GT].
	
	\bibitem[Phragm{\'e}n(1894)]{Phra94a}
	E.~Phragm{\'e}n.
	\newblock Sur une m{\'e}thode nouvelle pour r{\'e}aliser, dans les
	{\'e}lections, la repr{\'e}sentation proportionnelle des partis.
	\newblock \emph{\"Ofversigt af Kongliga Vetenskaps-Akademiens F\"orhandlingar},
	51\penalty0 (3):\penalty0 133--137, 1894.
	
	\bibitem[Phragm{\'e}n(1895)]{Phra95a}
	E.~Phragm{\'e}n.
	\newblock \emph{Proportionella val. En valteknisk studie}.
	\newblock Svenska sp{\"o}rsm{\aa}l 25. Lars H{\"o}kersbergs f{\"o}rlag,
	Stockholm, 1895.
	
	\bibitem[Phragm{\'e}n(1896)]{Phra96a}
	E.~Phragm{\'e}n.
	\newblock Sur la th{\'e}orie des {\'e}lections multiples.
	\newblock \emph{{\"O}fversigt af Kongliga Vetenskaps-Akademiens
		F{\"o}rhandlingar}, 53:\penalty0 181--191, 1896.
	
	\bibitem[Phragm{\'e}n(1899)]{Phra99a}
	E.~Phragm{\'e}n.
	\newblock Till fr{\aa}gan om en proportionell valmetod.
	\newblock \emph{Statsvetenskaplig Tidskrift}, 2\penalty0 (2):\penalty0
	297--305, 1899.
	
	\bibitem[Pukelsheim(2014)]{Puke14a}
	F.~Pukelsheim.
	\newblock \emph{Proportional Representation: Apportionment Methods and Their
		Applications}.
	\newblock Springer, 2014.
	
	\bibitem[S\'{a}nchez-Fern\'{a}ndez and Fisteus(2019)]{SF19}
	L.~S\'{a}nchez-Fern\'{a}ndez and J.~A. Fisteus.
	\newblock Monotonicity axioms in approval-based multi-winner voting rules.
	\newblock In \emph{Proceedings of the 18th International Conference on
		Autonomous Agents and Multiagent Systems (AAMAS)}, pages 485--493, 2019.
	
	\bibitem[S{\'a}nchez-Fern{\'a}ndez
	et~al.(2017{\natexlab{a}})S{\'a}nchez-Fern{\'a}ndez, Elkind, and
	Lackner]{SEL17a}
	L.~S{\'a}nchez-Fern{\'a}ndez, E.~Elkind, and M.~Lackner.
	\newblock Committees providing {EJR} can be computed efficiently.
	\newblock Technical report, arXiv:1704.00356v3 [cs.GT], 2017{\natexlab{a}}.
	
	\bibitem[S{\'a}nchez-Fern{\'a}ndez
	et~al.(2017{\natexlab{b}})S{\'a}nchez-Fern{\'a}ndez, Elkind, Lackner,
	Fern{\'a}ndez, Fisteus, {Basanta Val}, and Skowron]{SFF+17a}
	L.~S{\'a}nchez-Fern{\'a}ndez, E.~Elkind, M.~Lackner, N.~Fern{\'a}ndez, J.~A.
	Fisteus, P.~{Basanta Val}, and P.~Skowron.
	\newblock Proportional justified representation.
	\newblock In \emph{Proceedings of the 31st AAAI Conference on Artificial
		Intelligence (AAAI)}, pages 670--676. AAAI Press, 2017{\natexlab{b}}.
	
	\bibitem[S{\'a}nchez-Fern{\'a}ndez et~al.(2021)S{\'a}nchez-Fern{\'a}ndez,
	Fern{\'a}ndez, Fisteus, and Brill]{SFFB18a}
	L.~S{\'a}nchez-Fern{\'a}ndez, N.~Fern{\'a}ndez, J.~A. Fisteus, and M.~Brill.
	\newblock The maximin support method: An extension of the {D'Hondt} method to
	approval-based multiwinner elections.
	\newblock In \emph{Proceedings of the 35th AAAI Conference on Artificial
		Intelligence (AAAI)}, pages 5690--5697. AAAI Press, 2021.
	
	\bibitem[Skowron et~al.(2017)Skowron, Lackner, Brill, Peters, and
	Elkind]{SLB+17a}
	P.~Skowron, M.~Lackner, M.~Brill, D.~Peters, and E.~Elkind.
	\newblock Proportional rankings.
	\newblock In \emph{Proceedings of the 26th International Joint Conference on
		Artificial Intelligence (IJCAI)}, pages 409--415. IJCAI, 2017.
	
	\bibitem[Skowron et~al.(2016)Skowron, Faliszewski, and Lang]{SFL16a}
	P.~K. Skowron, P.~Faliszewski, and J.~Lang.
	\newblock Finding a collective set of items: From proportional
	multirepresentation to group recommendation.
	\newblock \emph{Artificial Intelligence}, 241:\penalty0 191--216, 2016.
	
	\bibitem[{Speroni di Fenizio} and Gewurz(2019)]{SpGe19a}
	P.~{Speroni di Fenizio} and D.~A. Gewurz.
	\newblock The space of all proportional voting systems and the most
	majoritarian among them.
	\newblock \emph{Social Choice and Welfare}, 52\penalty0 (4):\penalty0 663--683,
	2019.
	
	\bibitem[Thiele(1895)]{Thie95a}
	T.~N. Thiele.
	\newblock Om flerfoldsvalg.
	\newblock \emph{Oversigt over det Kongelige Danske Videnskabernes Selskabs
		Forhandlinger}, pages 415--441, 1895.
	
\end{thebibliography}

\clearpage
\appendix

\section{Omitted Proofs}\label{app:compAspects}

\subsection{Proof of \Cref{thm:pav_np_hard}}
\label{app:pavhard}

We show NP-hardness by reduction from the NP-complete problem \textsc{Independent Set}~\citep{GaJo79a}.
\problemdef{Independent Set}{Undirected graph $G=(V,E)$, $t \in \mathbb{N}$.}{Is there a vertex subset $V' \subseteq V$ of size $|V'|=t$ such that no two vertices in $V'$ are connected by an edge in $G$?}
This problem is NP-hard even when restricted to cubic graphs (where every vertex has degree~3) \citep{GaJo79a}. 
Our reduction is a simplified version of the reduction proposed by \citet[Theorem 1]{AGG+15a}.

\restatePAVhard*

\begin{proof}
For a given cubic graph $G=(V,E)$ and independent set size $t \in [|V|]$, we construct a \pa election $(\N,\P,\A,k)$ in the following.
    For each vertex $v \in V$, there is a party $p_v \in \P$.
    For every edge $e=\{u,v\} \in E$, there is one voter in $\N$ who approves exactly $p_u$ and $p_v$.
Lastly, we set $k=t$.

    This construction is clearly polynomial in the size of $G$.
    We show that $G$ has an independent set of size $t$ iff there is a committee $\W$ for the election $(\N,\P,\A,k)$ with $\PAV(\W) \geq s = 3 t$.

    ``$\Rightarrow$'': Assume that $G$ has an independent set $V' \subseteq V$ of size $|V'|=t$.
    Consider the committee $\W$ where for every vertex $v \in V'$, the party $p_v$ receives exactly one seat (thus, the committee has size $k=t$).
    Each party $p_v$ is approved by three voters, namely all those voters corresponding to edges that are incident to $v$.
    Because $V'$ is an independent set, no voter approves more than one party in the committee, and thus only has a single seat on the committee belonging to an approved party.
    Consequently, the total PAV score of $\W$ is exactly $3t$.

    ``$\Leftarrow$'': Assume that $\W$ is a committee with $\PAV(\W) \geq 3 t$ for the constructed election. The PAV score of a given committee can be computed by starting with the empty committee and then iteratively adding up the marginal PAV score of each seat within the committee. Every party $p_v$ is approved by exactly three voters and therefore, giving one seat to $p_v$ in the committee can increase the PAV score by at most three. As there are only $t$ seats available, every seat assignment has to increase the PAV score by exactly three. In order to achieve an increase of three when adding a seat to $p_v$, all the voters who approve $p_v$ must have been previously completely unrepresented. Thus, all parties present in $\W$ receive only one seat and do not have any common approving voters. By construction, this implies that the set of vertices $\{ v \in V : \W(p_v) > 0 \}$ corresponding to $\W$ is an independent set of size $t$. 
\end{proof}

\subsection{Proof of \Cref{thm:pav_comp_core}}
\label{subsec:efficientCore}

In the following we prove that in the \pa subdomain, core-stable committees can be computed in polynomial time. We make use of a local search procedure, introduced for the \ca setting by \citet{AEH+18a}, which approximates a local maximum of the PAV score function. \citet{AEH+18a} show that their algorithm runs in polynomial time and returns committees providing EJR. For \pa elections, we show that, by a minor adjustment of the algorithm, committees computed by LS-PAV satisfy core stability. In \Cref{alg:lspav} we slightly adjust the original definition by parameterizing the procedure by the approximation threshold $\epsilon$. Note that, once again, the algorithm is defined in terms of \ca elections; in \Cref{subsec:embeddingPappCapp} we show how to apply \ca rules to \pa elections. 

\begin{algorithm}
    \caption{LS-PAV (Candidate-Approval Rule)}
    \label{alg:lspav}
    \begin{algorithmic}
        \Function{LS-PAV}{$\N,\C,\A,k, \epsilon$}
            \State $\W \!\gets\! k$ arbitrary candidates from $C$ 
\While{$\exists \; c \in W, c' \in C \! \setminus \! W $ such that $PAV(W \! \setminus \! \{c\} \! \cup \! \{c'\}) \geq PAV(W) \!+ \! \epsilon$}
            \State $W \gets W \setminus \{c\} \cup \{c'\}$
            \EndWhile
            \State \Return \W
        \EndFunction
    \end{algorithmic}
\end{algorithm}

\restateCoreInPoly*

\begin{proof}
We show that LS-PAV with threshold $\epsilon = \frac{1}{(1 + 2 (k-1))(k-1)k}$
always selects a committee from the core when $k>1$ and that the procedure runs in polynomial time for this specific choice of $\epsilon$. This suffices to prove the Theorem as computing a core-stable committee for $k=1$ is trivial.
The proof is an extension of the proof of \Cref{thm:pav_papp_core}. 

For some \pa election $(N,P,A,k)$ let $W$ be a committee selected by LS-PAV with $\epsilon$ and assume that $W$ is not core stable. Hence, there exists $S \subseteq N$, $T: P \rightarrow \mathbb{N}$, $\ell \in [k]$ with $\vert S \vert \geq \ell \, n/k, \vert T \vert = \ell$ and 
\begin{equation}
u_i(T) > u_i(W) \; \forall \; i \in S. \label{eq:SisHappier}
\end{equation}
We start by reproducing some observations which were done within the beginning of the proof of \Cref{thm:pav_papp_core}. For more detailed arguments, we refer to this proof. 

We define the marginal contribution of a party $p$ to the PAV score of $W$
 \begin{equation*}
        \MC(p, \W) = \PAV(\W) - \PAV(\W - \{p\}) = \sum_{i \in \N \, : \, p \in A_i} \frac{1}{u_i(\W)}
    \end{equation*}
    and obtain an upper bound for the sum of the marginal contribution of all seats in $W$ w.r.t. $W$, i.e., \begin{equation}\sum_{p \in P} W(p) \cdot \MC(p,\W) \leq n \label{eq:lowerBoundAv}\end{equation} and a lower bound for the sum of marginal contribution of all seats in $T$ to $\W + \{p\}$, where $p$ is the party corresponding to the seat, \ie \begin{equation}\vert S \vert \leq \sum_{p \in P} T(P) \cdot \MC(p, \W + \{p\}).\label{eq:upperBoundAv}\end{equation}
Hence, there exists a party $p_1$ in the support of $W$ for which
\begin{equation}
\MC(p_1,W) \leq n/k \label{eq:lowerbound}
\end{equation}
holds and there exists a party $p_2$ in the support of $T$ for which 
\begin{equation}
n/k \leq \MC(p_2,W+\{p_2\})   \label{eq:upperbound}
\end{equation}
holds. We distinguish three cases: 

\begin{description}
\item[Case 1:]  It holds that
\begin{enumerate}
\item there exists $p_1 \text{ in the support of } W: \MC(p_1,W) \leq n/k - \epsilon$ \textbf{ or} \label{case1con1}
\item there exists $p_2 \text{ in the support of } T: n/k + \epsilon \leq \MC(p_2,W + \{p_2\})$. \label{case1con2}
\end{enumerate}

\smallskip

First assume that the first condition is satisfied. Then, let $p_1$ be a such a party from the support of $W$ and $p_2$ such that (\ref{eq:upperbound}) holds. We define the multiset $W' = W - \{p_1\} + \{p_2\}$ and observe that \begin{align*}
\PAV(W') &= \PAV(W) - \MC(p_1, W) + \MC(p_2, W - \{p_1\} + \{p_2\}) \\ 
& \geq \PAV(W) - \MC(p_1, W) + \MC(p_2, W + \{p_2\})  \geq\PAV(W) + \epsilon,
\end{align*}
a contradiction to the assumption that LS-PAV with parameter $\epsilon$ returned $W$. 
The last inequality follows from case condition \ref{case1con1} and (\ref{eq:upperbound}).

If the second condition holds, an analogous argument yields a contradiction. 

\smallskip

\item[Case 2:] It holds that 

\begin{enumerate}
\item for all $p_1 \text{ in the support of } W: \MC(p_1,W) > n/k - \epsilon$ \textbf{ and} \label{case2con1}
\item for all $p_2 \text{ in the support of } T: \MC(p_2,W\!+\!\{p_2\})\! <\! n/k\! +\! \epsilon$ \textbf{and} \label{case2con2}
\item there exists $i \in S: u_i(W) > 0$. \label{case2con3}
\end{enumerate}

Applying (\ref{eq:lowerBoundAv}) and case condition \ref{case2con1}, we obtain an upper bound for $\MC(p_1,W)$ for any $p_1$ in the support of $W$: 
\begin{align}
\MC(p_1,W) &\leq n - \sum_{p \in P \setminus{p_1}} W(p) \MC(p,W) -  (W(p_1)-1)\; \MC(p_1,W) \nonumber\\ 
&< n - (k-1)\left(\frac{n}{k}-\epsilon\right) = \frac{n}{k} + (k-1)\epsilon.\label{eq:upperbound2}
\end{align}

Analogously, applying (\ref{eq:upperBoundAv}) and case condition \ref{case2con2}, we obtain an lower bound for $\MC(p_2,W+\{p_2\})$ for any $p_2$ in the support of $T$. That is, 
\begin{align}
\MC(p_2,W+\{p_2\}) &\geq |S| - \sum_{p \in P \setminus{p_2}} T(p) \; \MC(p,W + \{p\}) -  (T(p_2)-1)\; \MC(p_2,W+\{p_2\}) \nonumber\\ 
&> |S| - (\ell-1)\left(\frac{n}{k}+\epsilon\right) \geq \frac{n}{k} - (k-1)\epsilon.\label{eq:lowerbound2}
\end{align}

Subsequently, choose some $i \in S$ with $u_i(W)>0$ (existence guaranteed by case condition \ref{case2con3}) and a party $p_1$ from the support of $W$ which is also included in the approval set of voter $i$, $A_i$. Then, choose a party $p_2$ in the support of $T$ which is also approved by voter $i$ but $W(p_2) < T(p_2)$ (existence guaranteed by the fact that voter $i$ prefers committee $T$ to committee $W$). Note that in particular, the restrictions made by case conditions \ref{case2con1} and \ref{case2con2} already imply that $p_1$ and $p_2$ are different parties.\footnote{Assume for contradiction that $p_1$ and $p_2$ are the same parties. Then, in particular it holds that $\MC(p_1,W) = \MC(p_2,W)$ and $\MC(p_1, W + \{p_1\}) = \MC(p_2, W + \{p_2\})$. Consider the difference $\MC(p_1, W) - \MC(p_1,W + \{p_1\})$. Note that we do not do any further assumptions in order to derive (\ref{eq:strictseparation}). Preempting (\ref{eq:strictseparation}), we know that $\MC(p_1, W) - \MC(p_1,W + \{p_1\}) \geq \frac{1}{k(k-1)}$ holds, but on the other hand we get from (\ref{eq:upperbound2}) and (\ref{eq:lowerbound2}) that $\MC(p_1,W) - \MC(p_1,W+\{p_1\}) \leq 2(k-1)\epsilon = \frac{1}{1/2 + k(k-1)} < \frac{1}{k(k-1)}$ holds, a contradiction.}

For this choice of $p_1$ and $p_2$ we aim quantify the gap between the contribution of $p_2$ with respect to $W - \{p_1\} + \{p_2\}$ and the contribution of $p_2$ with respect to $W + \{p_2\}$. More precisely, we will show that \begin{equation}
\MC(p_2,W - \{p_1\} + \{p_2\})  \geq  \MC(p_2,W + \{p_2\}) +  \frac{1}{k(k-1)}. \label{eq:strictseparation}
\end{equation}

To this end recall that voter $i$ supports party $p_1$ and hence
\begin{equation}
u_i(W - \{p_1\}) = u_i(W) - 1. \label{eq:refine1}
\end{equation}
Moreover, for all remaining voters $j \in N \setminus \{i\}$ it holds that
\begin{equation}
u_j(W - \{p_1\}) \leq u_j(W) \label{eq:refine2}.
\end{equation}
Lastly, from $i$ being in the deviator set $S$, we know that 
\begin{equation}
u_i(W) \leq k-1 \label{eq:refine3}.
\end{equation}

Let $N_p= \{i \in N : p \in A_i\}$ denote the set of supporters of party $p$ and $N_p^{-i} = N_p \setminus \{i\}$. Putting it all together, we get
\begin{align*}
\MC(p_2,W - \{p_1\} + \{p_2\}) &=  \sum_{j \in N_{p_2}} \frac{1}{u_j(W - \{p_1\}) + 1} \\ 
&=  \sum_{j \in N_{p_2}^{-i}} \frac{1}{u_j(W - \{p_1\}) + 1}  + \frac{1}{u_i(W - \{p_1\}) + 1}\\ 
& \geq  \sum_{j \in N_{p_2}^{-i}} \frac{1}{u_j(W)+ 1}  + \frac{1}{u_i(W)}\\ 
& =  \sum_{j \in N_{p_2}^{-i}} \frac{1}{u_j(W) +1}  + \frac{1}{u_i(W)+ 1} + \frac{1}{(u_i(W)+1)u_i(W)}\\ 
& \geq  \MC(p_2, W +\{p_2\}) +   \frac{1}{k(k-1)}.
\end{align*}
The first inequality holds due to (\ref{eq:refine1}) and (\ref{eq:refine2}) and the second due to (\ref{eq:refine3}). 

Finally, making use of (\ref{eq:upperbound2}),(\ref{eq:lowerbound2}), and (\ref{eq:strictseparation}), we can show 
\begin{align*}
 \PAV(W') &= \PAV(W) - \MC(p_1,W) + \MC(p_2,W - \{p_1\} + \{p_2\})   \\ 
&\geq \PAV(W) -  \MC(p_1,W) + \MC(p_2,W + \{p_2\}) + \frac{1}{k(k-1)} \\
&> \PAV(W) -\frac{n}{k} - (k-1)\epsilon + \frac{n}{k} -(k-1)\epsilon + \frac{1}{k(k-1)} \\ 
&= \PAV(W) -2(k-1)\epsilon + \frac{1}{k(k-1)} \\ 
&= \PAV(W) + \epsilon,
\end{align*}
a contradiction to the termination of LS-PAV. 
The first inequality is due to (\ref{eq:strictseparation}) and the second due to (\ref{eq:upperbound2}) and (\ref{eq:lowerbound2}).

\smallskip

\item[Case 3:] Finally, suppose that we are neither in Case~1 nor in Case~2. 
It follows that $\sum_{i \in S}u_i(W) = 0$ but $\sum_{i \in S}u_i(T) \geq |S| $. Hence, there exists some $p_2$ in the support of~$T$ with at least $\vert S \vert/\vert T \vert \geq n/k$ supporters in $S$. This is a contradiction to the fact that LS-PAV satisfies EJR which was shown by \citet{AEH+18a}.\footnote{Note that \citet{AEH+18a} show that LS-PAV satisfies EJR when $\epsilon' = \frac{n}{k^2}$. Since $\epsilon \leq \epsilon'$ for all $k \geq 2$, their result carries over to LS-PAV with $\epsilon$.}
\end{description}

Lastly, we show that LS-PAV for $\epsilon = \frac{1}{(1 + 2 (k-1))(k-1)k}$ runs in polynomial time in $|P|$,~$n$, and~$k$. 
We follow the proof by \citet{AEH+18a} showing that LS-PAV runs in polynomial time for $\epsilon' = n/k^2$. Per iteration of the while loop, the algorithm computes at most $mk$ PAV scores, which can be done in polynomial time. 
In order to bound the number of while loops, observe that the PAV score of a committee is upper bounded by $n H_k \in \mathcal{O}(n \ln{k})$ and the algorithm improves the PAV score of the best committee found so far in every iteration by at least $\epsilon$. Hence, there are $\mathcal{O}(nk^3\ln{k})$ iterations of the while loop, which suffices to prove the claim.   \end{proof}

\subsection{Checking PJR} \label{sec:checkPJR}

We start by defining \emph{proportional justified representation (PJR)} for \pa elections. 

\begin{definition}\label{def:pjr}
A committee $\W : \P \rightarrow \mathbb{N}$ provides \emph{proportional justified representation (PJR)}, if there is no $S \subseteq \N$ such that $\bigcap_{i \in S} A_i \neq \emptyset$ and $\sum_{p \in \bigcup_{i \in S}\! \A_i} W(p) < q(S)$.
\end{definition}

In words, PJR requires that for every voter group $S$ with a commonly approved party, the committee should contain at least $q(S)$ candidates from the union of all parties approved by voters in $S$. Observe that a committee providing EJR also provides PJR\@.

For showing that checking whether a committee satisfies PJR can be done in polynomial time, we use techniques from submodular optimization. Recall that, given a finite set $U$, a function $f : 2^U \to \mathbb{R}$ is \textit{submodular} if for all subsets $X, Y \subseteq U$ with $X \subseteq Y$ and for every $x \in U \setminus Y$, it holds that \[
        f(X \cup \{x\}) - f(X) \geq f(Y \cup \{x\}) - f(Y).\]
A submodular function $f : 2^U \to \mathbb{Z}$ can be minimized in time polynomial in $|U| + \log \max \{ |f(S)| \, : \, S \subseteq U \}$ \cite[Theorem 14.19]{KV18}.
Applying this result, one can check whether a \pa committee provides PJR in polynomial time.
\begin{theorem}
    Given a \pa election $(\N,\P,\A,k)$ and a committee $\W : \P \to \mathbb{N}$, it can be checked in polynomial time whether $\W$ satisfies PJR\@.
    \label{thm:party_pjr_check}
\end{theorem}

\begin{proof}
    We fix a committee $\W : \P \to \mathbb{N}$ and define the function $h : 2^{\N} \to \mathbb{N}$ by
    \[ h(S) = \sum_{p \in \bigcup_{i \in S}\! A_i} \W(p), \]
    i.e., for a voter group $S \subseteq \N$, $h(S)$ is the total number of seats that $\W$ allocates to to the union of all parties approved by voters in $S$. Moreover, for each party $p \in P$, we let $N_p = \{i \in N \mid p \in A_i\}$ denote the set of supporters of $p$. Observe that the committee $\W$ satisfies PJR if and only if there is no party $p \in P$ and group of voters $S \subseteq N_p$ with $h(S) < q(S)$. 

We show how to check in polynomial time for a fixed party $p \in P$, whether there exists such a group of voters $S \subseteq N_p$. Then, this procedure can be repeated for every party in $\P$.

We define the function $f : 2^{\N_p} \rightarrow \mathbb{R}$ by
    \[
        f(S) = h(S) - |S| \, \frac{k}{n}
    \]
 
    and show that $f$ is submodular. To this end let $X,Y \subseteq N_p$ with $X \subseteq Y$ and $x \in N_p \setminus Y$. Then, 
    \begin{align*}
    f(X \cup \{x\}) - f(X) &= \sum_{p \in A_x} W(p) - \sum_{p \in A_x \cap (\bigcup_{i \in X} \! A_i)} W(p) - \frac{k}{n} \\ 
    & \geq \sum_{p \in A_x} W(p) - \sum_{p \in A_x \cap (\bigcup_{i \in Y} \!A_i)} W(p) - \frac{k}{n}\\
    & = f(Y \cup \{x\}) - f(Y),
    \end{align*}
    which suffices to prove the submodularity of $f$.

    By multiplying $f$ by $n$, we obtain an integer-valued submodular function with $\max\{n\cdot |f(S)| : S \subseteq N_p\} \leq kn$; thus, we can minimize $f$ in time $\mathcal{O}(n + \log(kn))$.

We show in the following that any $S \subseteq \N_p$ is the witness of a PJR violation if and only if $f(S) \leq -1$.
    
For the direction from left to right,
    assume that $S \subseteq \N_p$ shows a violation of PJR, i.e., $h(S) < q(S)$. Since both values are integers, we know in particular that $h(S)\leq q(S)-1 = \big\lfloor |S|\frac{k}{n} \big \rfloor -1 \leq |S|\frac{k}{n} -1$ holds. This implies $f(S) \leq -1$. 
    
For the direction from right to left,
    fix some $S \subseteq \N_p$ with $f(S) \leq -1$. It follows that $h(S) \leq |S| \frac{k}{n} - 1 < q(S)$, a violation of PJR for the group~$S$.

The above observation implies a natural procedure to check for a PJR violating group within the supporters of some party $p$: Minimize the function~$f$ and check whether its minimum is larger than $-1$. If not, we have found a violation. If the minimum of $f$ is larger than $-1$ for all $p \in P$, then $\W$ satisfies PJR. The described algorithm runs in time $\mathcal{O}\big(|P|(n+\log(kn))\big)$.
\end{proof}

\section{Results on Further Multiwinner Voting Rules}
\label{sec:app:axiomatic}

In this section we consider other approval-based multiwinner voting rules from the literature and study their axiomatic properties in the \pa subdomain. Note that we use the language of the candidate-approval setting and in particular, $W$ is a set (not a multiset) of candidates. In order to apply the described rules in the \pa setting, we can transform any \pa election to a candidate-approval election by introducing $k$ clones of each party (see \Cref{subsec:embeddingPappCapp}).

We focus on five rules that satisfy PJR in the candidate-approval setting, and briefly comment on rules not satisfying PJR in \Cref{app:other}. For all five rules, we show that they do \textit{not} satisfy stronger proportionality axioms in the \pa subdomain; see \Cref{t:papp_summary} for a summary of our observations. Furthermore, we show that \maxP remains computationally intractable when restricting the domain to \pa elections.

\begin{table}[ht!]
\centering
\begin{tabular}{@{}lccccc@{}}
\toprule
Rule            & \phantom{xxx} PJR \phantom{xxx}     & \phantom{xxx} EJR \phantom{xxx}   & Core Stability           \\ \midrule
PAV               & \cmark & \cmark & \cmark \\
\seqP       & \cmark & \xmark & \xmark \\
\maxP      & \cmark & \xmark & \xmark \\
\eneP  & \cmark & \xmark & \xmark \\
Rule X               & \cmark & \cmark & \xmark \\
Maximin support method  & \cmark & \xmark & \xmark \\
\bottomrule
\end{tabular}

\caption{The table contains a summary of the axiomatic properties of candidate-approval rules within the subdomain of \pa elections.}
\label{t:papp_summary}
\end{table}

\subsection{\phrag's Rules}
\label{app:phragrules}
The first three rules we consider are (at least partially) due to Swedish mathematician Lars Edvard \phrag.\footnote{\phrag's original papers are written in French or Swedish \citep{Phra94a,Phra95a,Phra96a,Phra99a}; an English account of this work was composed by \citet{Jans16a}.} The first two rules, \maxP and \seqP, are based on the concept of \textit{load distributions}: It is assumed that adding a candidate to the committee incurs one unit of ``load,'' which needs to be distributed among the approvers of this candidate. The rules aim to select committees for which the associated load can be distributed as evenly as possible among the voters, where the balancedness of a load distribution is measured by the maximal total load of a voter. 

Formally, a real-valued vector $(x_{i,c})_{i \in \N, c \in \C}$ is a \emph{load distribution} for a candidate-approval election $(\N, \C, \A, k)$ if the following properties hold \citep{BFJL16a}:
\begin{align}
    &0 \leq x_{i,c} \leq 1 & \text{ for } i \in \N, c \in \C,
    \label{phr_def:1}\\
    &x_{i,c} = 0 & \text{ if } c \notin A_i,
    \label{phr_def:2}\\
    &\sum_{i \in \N} \sum_{c \in \C} x_{i,c} = k,
    \label{phr_def:3}\\
    &\sum_{i \in \N} x_{i,c} \in \{0,1\} & \text{ for } c \in \C.
    \label{phr_def:4}
\end{align}
In this definition, $x_{i,c}$ represents the load of candidate $c$ that is assigned to voter $i$. The total load of voter $i$ is given by $\sum_{c \in \C} x_{i,c}$. Properties (\ref{phr_def:3}) and (\ref{phr_def:4}) ensure that each load distribution corresponds to a committee of size $k$: candidate $c$ is in the committee if and only if $\sum_{i \in \N} x_{i,c} = 1$.

\medskip

The rule \textbf{\maxP} globally minimizes the balancedness of load distributions and returns committees corresponding to load distributions $(x_{i,c})$ such that $\max_{i \in \N} \sum_{c \in A_i} x_{i,c}$ is minimal. (Ties are broken in a leximax fashion; for details, we refer to \citet{BFJL16a}). In candidate-approval elections, \maxP satisfies EJR and is NP-hard to compute \citep{BFJL16a}. We first show that the computational intractability still holds for \pa elections. 

\begin{theorem}
    Computing a winning committee for \maxP is NP-hard in the \pa subdomain.
    \label{thm:maxphragmen_np_hard}
\end{theorem}
\begin{proof}
    The notion of load distributions can be adapted in a straightforward manner to \pa elections, by replacing constraint (\ref{phr_def:1}) with  $0 \leq x_{i,p} \leq k$ and constraint (\ref{phr_def:4}) with $\sum_{i \in N} x_{i,p} \in [k]$ for all $p \in P$.
    We prove that the following problem is NP-hard:
    \problemdef{Party Approval \maxP}
    {\pa election $(\N,\P,\A,k)$, distribution bound $s \in \mathbb{R}$}
    {Is there a load distribution $(x_{i,p})$ such that $\max_{i \in \N} \sum_{p \in A_i} x_{i,p} \leq s$?}

Similarly to the proof  of \cref{thm:pav_np_hard}, we use a polynomial reduction from \textsc{Independent Set} on cubic graphs. The reduction is a variant of the one by \citet{BFJL16a}, which shows that \maxP is NP-hard in the \ca setting.

    Given a cubic graph $G=(V,E)$ and independent set size $t \in \mathbb{N}$, we define the following \pa election: For every vertex $v \in V$, there is a party $p_v \in \P$. Additionally, for every edge $e=\{u,v\} \in E$, there is a voter in $\N$ who approves exactly $p_u$ and $p_v$. The committee shall be as large as the independent set, that is, $k=t$. To prove that this reduction is sound, we show that $G$ has an independent set of size $t$ if and only if there is a load distribution $(x_{i,p})$ with $\max_{i \in \N} \sum_{p \in A_i} x_{i,p} \leq \frac{1}{3}$.

    ``$\Rightarrow$'': Assume $G$ has an independent set $V' \subseteq V$ of size $|V'|=t$. Because $G$ is cubic, every party in the created election is approved by exactly 3 voters. We define a valid load distribution, in which every party corresponding to a vertex in $V'$ creates a load of $\frac{1}{3}$ on every approving voter. (This also implies that in the induced committee, the parties corresponding to $V'$ receive exactly one seat.) Because $V'$ is an independent set, no voter receives load from multiple parties, and hence the maximal total load of every voter is $\frac{1}{3}$.

    ``$\Leftarrow$'': Assume there is a load distribution $(x_{i,p})$ such that $\max_{i \in \N} \sum_{p \in A_i} x_{i,p} \leq \frac{1}{3}$. Since every party is approved by exactly 3 voters, it follows that $x_{i,p} = \frac{1}{3}$ for a voter $i$ who approves a party $p$ that receives a seat in the induced committee. Consequently, no party receives more than one seat in the induced committee and no voter approves more than one party in the committee. Thus, the committee induces the independent set $\{ v \in V : x_{i, p_v}\! >\! 0 \text{ for some }i \in \N \}$ of size $t$. 
\end{proof}

In order to prove that \maxP does not satisfy EJR in the \pa setting, we use straightforward adaptation of an example by \citet{ABC+16a} (which is also used by \citet{SFF+17a} and \citet{BFJL16a}). 

\begin{proposition}\label{prop:maxP-ejr}
\maxP does not satisfy EJR for \pa elections. 
\end{proposition}

\begin{proof}
    Let $n=8$, $k=4$, and $P=\{A,B,C,D,X\}$. 
The ballot profile is given by
    \begin{alignat*}{4}
    &1 \times \{A,X\}, \qquad
    &1 \times \{B,X\}, \qquad
    &1 \times \{C,X\}, \qquad
    &1 \times \{D,X\}, \\
    &1 \times \{A\}, 
    &1 \times \{B\}, \qquad
    &1 \times \{C\}, 
    &1 \times \{D\}.
    \end{alignat*}
In this election, \maxP gives one seat each to the parties $A,B,C,D$ and thus achieves a perfectly balanced load distribution. Consider the group consisting of the four voters approving party $X$. This group has a quota of $2$, but no voter in this group is represented twice in the \maxP committee.  
    
\end{proof}

The instance from the proof of \Cref{prop:maxP-ejr} also shows that the incompatibility of EJR and \textit{proportional representation (PR)}, a proportionality axiom proposed by \citet{SFF+17a}, remains intact in the \pa subdomain.

\medskip

The rule \textbf{\seqP} constructs committee sequentially, starting with the empty committee and iteratively adding a candidate that increases the maximum voter load the least.
For a formal definition, we again refer to \citet{BFJL16a}.
Seq-\phrag does not satisfy EJR in candidate-elections, and the same is true for the \pa subdomain. 

\begin{proposition}
    \seqP fails EJR in \pa elections.
    \label{lm:seqphragmen_ejr_papp}
\end{proposition}
\begin{proof}
    Fix a natural number $k \geq 282$. We construct a \pa election with parties $A, B, C, D, E, X$. The ballot profile of the $n=2k$ many voters is as follows:
    \begin{align*}
    &1 \times \{A,X\}, \qquad
    1 \times \{B,X\}, \qquad
    1 \times \{C,X\}, \qquad
    1 \times \{D,X\}, \\
    &7 \times \{A,B,C,D\},  \hspace{3.6cm} (2k-11) \times \{E\}.
    \end{align*}
    We first ignore the voters approving $E$ and focus on the $11$ remaining voters. Initially, adding a seat to $A,B,C,$ or $D$ would increase the maximal voter load to $\frac{1}{8}$, while giving $X$ one seat would increase it to $\frac{1}{4}$. Without loss of generality, assume $A$ receives this seat. Then, giving the next seat to $B,C,$ or $D$ would increase the maximal load to $(\frac{7}{8}+1) \cdot \frac{1}{8}=\frac{15}{64}$; giving it to~$A$ would increase it to $\frac{1+1}{8} = \frac{16}{64}$, and giving the seat to~$X$ would increase it to $\frac{\frac{1}{8}+1}{4}=\frac{18}{64}$. Thus, we can assume~$B$ receives the seat. Analogously, the next two seats are allocated to~$C$ and~$D$, respectively---the exact computations can be found in \Cref{t:seqphragmen}. The fifth seat would then be allocated again to $A$, increasing the maximal voter load to $\frac{16473}{32768} \approx 0.50272$.

\begin{table}[t]
\centering
\begin{tabular}{@{}llllll@{}}
\toprule
Party & Iteration 1              & Iteration 2                & Iteration 3                & Iteration 4                & Iteration 5                \\ \midrule
A & \textbf{0.125} & 0.25             & 0.34570          & 0.42944          & \textbf{0.50272} \\
B & 0.125          & \textbf{0.23438} & 0.35938          & 0.44312          & 0.51639          \\
C & 0.125          & 0.23438          & \textbf{0.33008} & 0.45508          & 0.52835          \\
D & 0.125          & 0.23438          & 0.33008          & \textbf{0.41382} & 0.53882          \\
X & 0.25           & 0.28125          & 0.33984          & 0.42236          & 0.52582          \\ \bottomrule
\end{tabular}
\caption{The \seqP computation for the profile in \Cref{lm:seqphragmen_ejr_papp}, when party~$E$ is ignored. The table shows, for each iteration, the maximal voter load that would result from assigning the next seat to a given party, rounded to five significant digits. The bold entries denote which party receives a seat (with lexicographic tie-breaking).}
\label{t:seqphragmen}
\end{table}

Taking the voters for $E$ into account would not affect the computation above, because all voters who approve $E$ do not approve any other party. Thus, in every \seqP iteration, either~$E$ receives a seat or one of $A,B,C,D$ receives a seat, until $A,B,C,D$ all have one seat. Adding a seat to~$E$ increases the load of a voter approving $E$ by $\frac{1}{2k-11}$. Thus, if $k-4$ seats are allocated to $E$, each $E$-voter would have a load of $\frac{k-4}{2k-11}$. Observe that $\lim_{k\to\infty} \frac{k-4}{2k-11} = \frac{1}{2}$ and indeed, $\frac{k-4}{2k-11} < 0.50272$ for all $k\geq282$. Therefore, \seqP returns a committee where $A,B,C,D$ each receive one seat and $E$ receives the remaining $k-4$ seats.

    This is a contradiction to EJR: Since $n/k=2$, EJR demands one of the four voters approving $X$ to be represented at least twice in the committee. This is not the case for the committee selected by \seqP. 
\end{proof}

\medskip

Additionally, Phragmén developed a voting rule which adapts the well-known single transferable vote (STV) system to approval ballots. Following \citet{CMS19a}, we refer to this method as \textbf{\eneP}.
Like \seqP, this method selects candidates iteratively. Initially, every voter has \textit{weight} of $1$ and a candidate's score is the sum of all approving voters' weights. In every round, the candidate with the highest score $s$ is added to the committee. If a voter $i$ with weight $f_i$ approves the candidate who is added to the committee, then their weight will be updated to $f_i \cdot (s - n/k) / s$ if $s > n/k$, and to $0$ otherwise. This process is repeated until all $k$ seats are assigned.

\begin{table}[t]
\centering
\begin{tabular}{@{}llllllll@{}}
\toprule
Party & Round 1 & Round 2 & Round 3 & Round 4 & Round 5 & Round 6 & Round 7 \\ \midrule
$A$ & 240.0 & 179.86 & 179.86 & 103.26 & 103.26 & 103.26 & \textbf{103.26}  \\
$X_1$ & \textbf{242.0} & 120.71 & 120.71 & 120.71 & 120.71 & 120.71 & 60.14  \\
$X_2$ & 190.0 & 190.0 & \textbf{190.0} & 68.71 & 45.96 & 45.96 & 45.96  \\
$X_3$ & 183.0 & 121.86 & 121.86 & 121.86 & 121.86 & \textbf{121.86} & 0.57  \\
$X_4$ & 240.0 & 240.0 & 179.61 & \textbf{134.92} & 13.64 & 13.64 & 13.64  \\
$X_5$ & 241.0 & \textbf{241.0} & 119.71 & 119.71 & 66.13 & 7.86 & 7.86  \\
$X_6$ & 186.0 & 186.0 & 125.11 & 125.11 & \textbf{125.11} & 3.82 & 3.82  \\
\bottomrule
\end{tabular}
\caption{The \eneP computation for the restricted profile in \Cref{lm:p1_ejr_papp}. The table shows the scores of the parties in the first seven iterations. The bold entries denote the party with the highest score.}
\label{t:phragmenstv}
\end{table}

The \eneP rule does not satisfy EJR in candidate-approval elections \citep{SEL17a,CMS19a}, and the same holds for \pa elections. 

\begin{proposition}
    \eneP fails EJR in \pa elections.
    \label{lm:p1_ejr_papp}
\end{proposition}
\begin{proof}
    For $k \geq 18$, consider an election with $n=120k$ voters and $k+1$ parties $A, X_1, \ldots, X_{k}$. The ballot profile is as follows:
\begin{alignat*}{3}
        &120 \times \{A, X_1\}, \qquad
        &120 \times \{A, X_2\}, \qquad
        &122 \times \{X_1, X_3\}, \\
        &70 \times \{X_2, X_4\}, \qquad
        &120 \times \{X_4,X_5\}, \qquad
        &121 \times \{X_5,X_6\}, \\
        &61 \times \{X_3\}, \qquad
        &50 \times \{X_4\}, \qquad
        &65 \times \{X_6\}, \\
        &109 \times \{X_{j}\} \;\text{ for }j \in \{7,\ldots,15\},\\
        &110 \times \{X_{j}\} \; \text{ for }j \in \{16,17,18\}.
    \end{alignat*}
    If $k>18$, we also add 120 voters approving $\{X_j\}$ for every $j \in \{19, \ldots, k\}$.

    First, consider the parties $A,X_1, \ldots, X_6$ only. In \Cref{t:phragmenstv}, the \eneP calculation for an election restricted to these parties is described. Note that in the first 6 iterations, the parties $X_1, \ldots, X_6$ receive one seat each and all have, when selected as winners, a score that exceeds 120. Afterwards, every party has a score strictly smaller than~109.
    
    Furthermore, observe that the parties $X_7, \ldots, X_k$ are all approved by voters who only approve this one particular party. As a result, their scores are not affected when other parties receive a seat. The parties $X_7, \ldots, X_{15}$ have a score of 109, $X_{16},X_{17},X_{18}$ a score of 110, and $X_{19}, \ldots, X_k$ (if they exist) a score of 120. When any of these parties receive a seat, their score is decreased to 0, as they are all approved by at most $n/k$ voters.

    Together, this shows that \eneP firstly allots one seat each to $X_1, \ldots, X_6$. Then, the score of the parties $A, X_1, \ldots, X_6$ is always smaller than 109, and therefore, $X_7, \ldots, X_k$ all receive a seat, which fills the committee. Thus, in the committee selected by \eneP, $X_1, \ldots, X_{k}$ each receive one seat. However, the $240=2 n/k$ voters who approve $A$ form a cohesive group, where at least one voter should be represented by at least two seats according to EJR\@. This is not the case in the committee generated by \eneP\@.  
\end{proof}

\subsection{Rule X}

\textit{Rule X} has been proposed by \citet{PeSk20awitharxiv}. Rule~X is similar to \seqP, but satisfies stronger proportionality guarantees. In particular, Rule~X satisfies EJR, but not core stability \citep{PeSk20awitharxiv}. The same holds for the \pa setting.

\begin{proposition}
    Rule X fails core stability in \pa elections.
\end{proposition}
\begin{proof}
    Consider again the \pa election from Example~\ref{ex:italiannocore}. 
In this election, Rule X gives 8 seats to party $p_0$ and 4 seats each to parties $p_2$ and $p_4$.
    As explained in Example~\ref{ex:italiannocore}, this committee is not in the core.

\end{proof}

\subsection{Maximin Support Method}

The \textit{maximin support method (MMS)} has been proposed by \citet{SFFB18a}. The method has strong similarities to \seqP and selects candidates sequentially. \citet{SFFB18a} show that MMS satisfies PJR, but not EJR. We adapt their EJR counterexample to the \pa setting. 

\begin{proposition}
    The maximin support method fails EJR in \pa elections.
\end{proposition}
\begin{proof}
    Let $n=8$, $k=4$, and $P=\{A,B,C,X\}$. The ballot profile is given by
    \begin{alignat*}{4}
    &5 \times \{A,X\}, \qquad
    &4 \times \{B,X\}, \qquad
    &3 \times \{C,X\},  \\
    &2 \times \{A\}, \qquad
    &1 \times \{B\}, \qquad
    &1 \times \{C\}.
    \end{alignat*}
    In this election, MMS gives one seat each to the parties $A,B,C,X$. Consider the group consisting of the 12 voters approving party $X$. This group has a quota of $3$, but no voter in this group is represented three times in the \maxP committee.  
\end{proof}

\subsection{Other Rules}
\label{app:other}

We also considered several other rules from the literature that are known to violate PJR in the candidate setting:
sequential PAV and reverse sequential PAV~\citep{Thie95a,Jans16a}, 
satisfaction approval voting~\citep{BrKi14a}, 
minimax approval voting~\citep{BKS07a}, 
var-Phragmén~\citep{BFJL16a}, 
GreedyMonroeAV~\citep{SFF+17a}, 
Approval Voting, 
MonroeAV, 
GreedyAV, 
HareAV, and 
Chamberlin--CourantAV (for definitions of the latter five rules, we refer to the article by \citealp{ABC+16a}).
For each of these rules, we verified that they do not satisfy PJR in the \pa subdomain either.  
Since existing counterexamples can be easily adjusted to the \pa setting, we omit the details.

\end{document}